%% file: main.tex
\pdfoutput=1
\documentclass[11pt,a4paper]{article}
\usepackage[margin=1in]{geometry}
\usepackage{float}
\usepackage{multitool}
\defintervals\defeps\defps

\usepackage{authblk}
\usepackage{hyperref}

\usepackage{braket}
\usepackage{cleveref}
\usepackage{mathrsfs}
\usepackage{tikz}
\usetikzlibrary{cd}

\usepackage{nicematrix}
\usepackage{enumerate}

\title{A Quantum Pigeonhole Principle and Two \\Semidefinite Relaxations of Communication Complexity}
\author[1]{Pavel Dvo\v rak}
\author[2]{Bruno Loff}
\author[2]{Suhail Sherif}
\affil[1]{Charles University, Prague}
\affil[2]{LASIGE and Faculty of Sciences of the University of Lisbon}

\date{}

\input{macros}

\usepackage{todonotes}

\begin{document}

\maketitle




\begin{abstract}
    We are interested in what happens when we take a $\Pi_1$ combinatorial statement, write its negation as a homogeneous quadratic feasibility problem (HQFP) (which is always possible since they are NP-complete), and relax the problem into a positive semidefinite feasibility problem. This question is particularly interesting owing to the fact that any statement written as a PSD feasibility problem can be proven or disproven using a short proof. We investigate this for one very simple and one very complicated statement.

    We start with the pigeonhole principle, writing its negation as a particular HQFP, and taking the PSD relaxation. We prove that this relaxed negation of the PHP, which in principle could be easier to satisfy, remains unsatisfiable, and we thus obtain a new ``quantum'' pigeonhole principle (QPHP) which is a stronger statement than the vanilla PHP. The QPHP states that if we take $n$ copies of the same state, and measure each copy using a measurement with only $n-1$ outcomes (the measurement can be different for different copies), then there will be an outcome $j$ and two copies $i_1, i_2$ where the resulting states, obtained when the outcome is $j$ for both copies, are not orthogonal.

    We then work with the statement ``the deterministic communication complexity of $f$ is $\le k$'', where $f$ could be either a function or a relation. We write this statement in two equivalent ways, using two different HQFPs. By relaxing to PSD feasibility, we increase the set of available protocols, and thus we always get a communication model which is stronger than deterministic communication complexity. It can be shown, by an argument from proof complexity, that any model obtained in this way will solve all Karchmer--Wigderson games efficiently. However, the details of how this happens are not at all clear: the argument is very indirect and does not give us an explicit protocol in the new model. We then work to find such protocols in the two communication models obtained by relaxing our two formulations.

    When relaxing the first of the two formulations, we obtain a kind of \textit{structured} variant of the $\gamma_2$ norm. This communication model is to matrices with subunit $\gamma_2$ norm like deterministic protocols are to rectangles, and so  we call \textit{$\gamma_2$ protocols} to the protocols in this model. We show that log-inverse-discrepancy is a lower-bound for this model, so, e.g., inner-product-mod-2 is a hard function in the model. We then show how to compute equality (deterministically) using $O(1)$ bits of $\gamma_2$-communication, which implies that KW games are easy in the model.

    When relaxing the second of the two formulations, we obtain a communication model, which we call \textit{quantum lab protocols}. This model happens to have a functional description, as follows. Alice is given $x$, Bob is given $y$, and they have access to a quantum lab where they have prepared some quantum system in an initial state $\psi_0$ (independent of $x$ and $y$). Then Alice and Bob take turns going to the lab, at each turn interacting with the quantum system by performing a single measurement, and writing down the outcome in the lab's whiteboard. The outcome of the last measurement should be $f(x,y)$ (with zero error probability). We use the QPHP to prove a lower-bound of $n$ against two-round quantum lab protocols for equality. We expected this to generalize to any number of rounds, but we ultimately show that \emph{any} Boolean function $f$ can be computed in three rounds and four measurements.
\end{abstract}

\newpage
\tableofcontents

\input{00_intro}
\input{01_preliminaries}

\input{02_qphp}
\input{03_sgamma2}

\input{04_quantumlab}
\input{05_nogo}

\input{06_conclusion}

\addcontentsline{toc}{section}{Acknowledgements}
\section*{Acknowledgements}

The authors would like to thank Carlos Florentino for fun conversations around this topic.

This work was funded by the European Union (ERC, HOFGA, 101041696). Views and opinions expressed are however those of the author(s) only and do not necessarily reflect those of the European Union or the European Research Council. Neither the European Union nor the granting authority can be held responsible for them.
It was also supported by FCT through the LASIGE Research Unit, ref.\ UIDB/00408/2020 and ref.\ UIDP/00408/2020, and by CMAFcIO, FCT Project UIDB/04561/2020, \url{https://doi.org/10.54499/UIDB/04561/2020}.
P. Dvořák was supported by Czech Science Foundation GAČR grant \#22-14872O.
\bibliographystyle{alpha}
\addcontentsline{toc}{section}{References}
\bibliography{main}



\end{document}

%% file: macros.tex
\newcommand{\R}{\mathbb{R}}
\newcommand{\N}{\mathbb{N}}
\newcommand{\ip}[2]{\langle #1,#2 \rangle}

\newmathsfs{Alice,Bob}
\newcommand{\DCC}{\mathsf{D}^{\mathsf{cc}}}

\newcommand{\GDCC}{\mathsf{\Gamma_2 D}^{\mathsf{cc}}}
\newcommand{\GLCC}{\mathsf{\Gamma_2 L}^{\mathsf{cc}}}
\newcommand{\disc}{\mathsf{disc}}

\newmathcals{A,K,H,M,D,T,X,Y,Z,P,L} 
\newmathsfs{lspan:span,affine,PSD,CP,tr:t,NP,coNP,poly,PTIME:P,EQ,NC,cl,KW}

\newtheorem{theorem}{Theorem}
\newtheorem{lemma}[theorem]{Lemma}
\newtheorem{claim}[theorem]{Claim}
\newtheorem{corollary}[theorem]{Corollary}
\newtheorem{definition}[theorem]{Definition}

\newtheorem*{remark*}{Remark}

\numberwithin{theorem}{section} 

%% file: 00_intro.tex
\newcommand{\mysubsubsection}[1]{%
\phantomsection
\addcontentsline{toc}{subsection}{#1}%
\subsubsection*{#1}%
}
\newpage
\section{Introduction}

The good thing about $\Sigma_1$ statements is that proving them amounts to finding a witness, after which the proof is a routine verification. But---if we assume that $\NP \neq \coNP$---there will necessarily exist $\Pi_1$ statements which cannot be proven in this way. Simultaneously, there exists a small number of situations when a particular class of $\Sigma_1$ statements is closed under negation, meaning, every statement in this class can be either proven or disproven by finding an explicit, easy-to-verify witness. Of course, this includes all ``easy'' statements (decidable in $\PTIME$), but beyond that the exhaustive list is quite short: conic feasibility, which includes semidefinite feasibility, (approximate) lattice problems, and stochastic games. To our knowledge, these three families of problems include all problems that are known to be in $\NP \cap \coNP$,\ 
\footnote{More precisely, conic feasibility is known to be in $\NP(\R) \cap \coNP(\R)$, as there are issues with the bitlength of solutions, which appear unavoidable. For example, one can construct a semidefinite feasibility problem $(\cA,b)$, with polynomially-many bits of of precision, which is satisfiable, but any solution $x$ must be specified with exponentially-many bits of precision \cite{khachiyan2000integer}.} 
but not known to be in $\PTIME$.
In this paper, we focus on semidefinite feasibility problems (SDFPs), which are a particular kind of conic feasibility, although similar considerations could be made for lattice problems and stochastic games.

\mysubsubsection{Homogenous Quadratic and Semidefinite Feasibility Problems (HQFP and SDFP)} In a linear feasibility problem, we are given a linear map $\cA:\R^n\to\R^m$ and a vector $b \in \R^m$, and we wish to know if there exists $x \in \R^n_{\ge 0}$ such that $\cA(x) = b$. As it turns out, many (but not all) of the properties of linear programming generalize to the case where the non-negative orthant $R^n_{\ge 0}$ is replaced by a closed, convex cone $\cK$, namely, a subset of $\R^n$ closed under limits, sums and multiplication by non-negative scalars.

In a semidefinite feasibility problem, we are given a linear map $\cA:\R^{\frac{n(n+1)}{2}}\to\R^m$ from the set of all symmetric matrices to $\R^m$, and a vector $b \in \R^m$, and we wish to know if there exists a positive semidefinite matrix $M$ such that $\cA(M) = b$. I.e., we replace the non-negative orthant $\R^n_{\ge 0}$ with the cone of \textit{positive semidefinite} $n\times n$ matrices $\PSD_n \subseteq \R^{\frac{n(n+1)}{2}}$ (such matrices are symmetric). This set can be alternatively characterized as the set of symmetric matrices with non-negative eigenvalues, or as the set of \textit{Gram} matrices, i.e., matrices equal to $A^\tr A$ for some $n\times m$ matrix $A$, or in other words, matrices $M$ of inner products, given by a family of vectors $a_1, \ldots, a_n$ (the columns of $A$), so that $M_{i j} = \braket{a_i \mid a_j}$.

It follows that a SDFP is asking whether there exist vectors $a_1, \ldots, a_n$ obeying a given system of linear equations on their inner products $\braket{a_i \mid a_j}$. (With linear programming being the special case where the linear equations only depend on the diagonal entries of $M$.) One can easily show that the dimension $m$ can be made to be $\le n$. Hence, the matrix $A$ serves as a short, easy-to-verify witness that a given SDFP is feasible.

Now, suppose we further restrict the solution $M$ to have rank $1$, i.e., the vectors $a_i$ and $a_j$ are now scalars. We then obtain a system of linear equations on degree-2 products $a_i \cdot a_j$, and we wish to know if some choice of scalars satisfies these equations. This is a different kind of problem, called a Homogeneous Quadratic Feasibility Problem (HQFP), and it is easily shown to be NP-hard.

\mysubsubsection{Being Relaxed about the Truth Helps in Finding Short Proofs}
It then follows that it is possible to take any $\Sigma_1$ combinatorial statement $\Psi$, write it down as a HQFP $Q$, and then relax it by droping the rank-$1$ restriction, to obtain a SDFP $P$. 

A radical transformation always happens in this process. The statement ``$Q$ is feasible'' is equivalent to $\Psi$, and by relaxation it always implies ``$P$ is feasible''. However, there is a fundamental result of Ramana \cite{ramana1997exact} saying that given any SDFP $P$ we can efficiently construct a different ``dual'' SDFP $P'$, such that ``$P$ is \textit{not} feasible'' if and only if ``$P'$ is feasible''. Hence, if $P$ is not feasible, we can always prove that $P$ is not feasible by presenting a short, easy witness --- the witness that $P'$ is feasible. So if $\Psi$ is true, ``$P$ is feasible'' remains true, and if $\Psi$ is false, then either ``$P$ is feasible'' becomes true (we relaxed too much), or ``$P$ is feasible'' is also false. In the latter case, there exists a short, easy witness that proves ``$P$ is not feasible'', and hence also proves that $\Psi$ is false. In other words, the relaxation map sends instances of an $\NP$-complete problem to instances of a problem in $\NP \cap \coNP$. Understandably, then, not all false $\Sigma_1$ statements $\Psi$ will remain false after relaxation, but when they do, we are guaranteed to have short proofs of falsity. 

\medskip
Now suppose there exists a particular $\Pi_1$ statement $\Psi$ we wish to prove. Maybe it is a tautological combinatorial principle, or even a complexity lower-bound. We then write $\neg\Psi$ as a HQFP $Q$ and relax it into the SDFP $P$ and try to prove that $P$ is false by constructing a solution for $P'$. If we succeed, it then follows that $\neg\Psi$ is false, i.e., $\Psi$ is true, and this is witnessed by a short, easy-to-verify object. Or maybe, encouraged by the guaranteed existence of a short proof of $P'$, we may try to prove that $P$ is false in another way, without necessarily aiming for a ``canonical'' proof.

In this paper, we report on what happens when we carry out the above approach, for two different $\Pi_1$ statements: the pigeonhole principle, and communication complexity lower-bounds. The whole approach can be seen as trying to express $\Pi_1$ statements in a very simple proof system, and we will have more to say below on the connection with proof complexity.

\mysubsubsection{The Quantum Pigeonhole Principle}

We formalize the negation of the pigeonhole principle (PHP) as a HQFP in a way similar to what has been done before in the polynomial calculus proof system (e.g. \cite{razborov1998lower}), by having $n m$ variables $v_{ij}$, indicating whether pigeon $i$ went to hole $j$, requiring that $\sum_j v_{ij}^2 = 1$, $v_{i j} \cdot v_{i j'} = 0$ (pigeon $i$ does not go into two holes) and $v_{i j} \cdot v_{i' j} = 0$ (no two pigeons go to the same hole). A small difference to the previous formalization is required so that the program is homogenous, but the crucial difference is that we then relax the homogenous quadratic program to a semidefinite program. The quantum pigeonhole principle (QPHP) is then the negation of this relaxed negation of the PHP, and therefore it is necessarily a stronger statement, i.e. it implies the PHP. 

In the language of linear algebra, the QPHP states the following. Suppose we take a unit vector $\lambda$ and decompose it into $h$ orthogonal vectors, in $p$ different ways:
\begin{align*}
    \sum_{j=1}^h v_{i,j} & = \lambda \tag{$i = 1, \ldots, p$}\\
    \braket{v_{i, j} \mid v_{i, j'}} & = 0 \tag{$\forall i \forall j \neq j'$}
\end{align*}
(i.e. we orthogonally distribute each of $p$ equal ``pigeons'' among $h$ ``holes''). Then if $h < p$, there will always exist a ``hole'' $j \in [h]$ and two ``pigeons'' $i \neq i'$, such that $\braket{v_{i, j} \mid v_{i', j}} \neq 0$

It is also possible to state the QPHP using only quantum language, as follows. Suppose that we have $p$ quantum registers $1,\ldots, p$, which are all initialized in the same state: $\ket{\psi_1} = \cdots = \ket{\psi_p}$. We then apply an $h$-outcome measurement to each of the registers. The specific measurement which we make may be different for different registers. Regardless, the measurements cause the registers to collapse to possibly-different states $\ket{\phi_1}, \ldots, \ket{\phi_p}$. The QPHP states that, if $h < p$, there will always exist an outcome $j$ and two registers $i \neq i'$, such that there is a non-zero probability of obtaining the same outcome $j$ after measuring both registers $i$ and $i'$, and when this happens the resulting states $\ket{\phi_i}$ and $\ket{\phi_{i'}}$ are not orthogonal.

In \Cref{sec:qphp}, we prove three versions of this statement. In \Cref{sec:qphp-AM-GM}, we prove it using the AM-GM inequality. The proof is short and simple, but can only show non-orthogonality if the number of holes $h$ is significantly smaller than the number $p$ of pigeons, namely $h < \frac{1}{4} \sqrt p$. In \ref{sec:qphp-geometric}, we prove a strong generalization of the QPHP, which allows for the initial states $\psi_i$ to be different for different pigeons, and gives a tight lower-bound on the maximal overlap $\braket{\phi_i \mid \phi_{i'}}$, as a function of the average initial overlap $\frac{1}{p (p-1)} \sum_{i \neq i'} \braket{\psi_i \mid \psi_j}$. Finally, in \Cref{sec:qphp-dual}, we provide one of the short ``canonical'' proofs which are guaranteed to exist via duality. Namely, we derive a feasibility problem dual to the relaxed negation of the QPHP, and give an explicit solution for it.

\mysubsubsection{Connection with Natural Proofs and Proof Complexity}

\Cref{sec:sgamma2,sec:quantumlab} of the paper apply the above approach to statements of the form ``the communication complexity of $f$ is $> k$''. This is a $\Pi_1$ statement when the two-player function (or relation) $f$ is given as a communication matrix. Indeed, the statement ``the communication complexity of $f$ is $\le k$'' is easily seen to be $\Sigma_1$, by taking an existential quantifier over all protocols. 

When starting this project over two years ago, our naive hope was that maybe we could use semidefinite programming to prove some new lower-bounds against Karchmer--Wigderson games. This would follow a long, successful tradition of using convex optimization to prove lower-bounds: approximate and threshold degree \cite{bun2022approximate}, the quantum adversary bound \cite{li2021general}, and the $\gamma_2$ norm \cite{{LMSS07}} are all examples of complexity measures which relax classical measures in one way or another, and which have been used to prove lower-bounds on classical and quantum query complexity, communication complexity, proof complexity, data structures, \textit{etc}. 

But also, such attempts have systematically failed against more powerful computational models, such as Boolean circuits and formulas. A famous result by Karchmer, Kushilevitz and Nisan \cite{karchmer1995fractional} (CCC'92) shows that the smooth partition bound is small for every Karchmer--Wigderson relation.\footnote{This result was generalized by Hrubeš et al. \cite{hrubevs2010convex}, to show that any ``convex rectangle measure'' assigns small complexity to KW relations.} A smooth partition is a linear-programming relaxation of an integer program defining the partition number, which is the smallest number of monochromatic rectangles needed to partition a communication matrix, itself a relaxation of the number of leaves in a communication protocol. KKN were hoping \cite[page 2]{karchmer1995fractional} that such a linear relaxation would help them prove lower bounds on the communication complexity of Karchmer--Wigderson relations, and hence lower bounds on the depth of Boolean formulas. Sadly, they could only report on a failed attempt. A few years later, Razborov and Rudich presented their natural-proofs barrier \cite{razborov1997natural} (STOC'95), which strongly suggests that no linear programming relaxation, or any other efficiently computable quantity, will be able to approximate the computational complexity of any model which is powerful enough to contain pseudorandom function generators.

One might think that the natural proofs barrier applies here, but one would be subtly mistaken. Indeed, semidefinite feasibility is not known to be in $\PTIME$, and there is significant evidence that it is actually a hard problem \cite{tarasov2008semidefinite}\footnote{We are referring to a result by Tarasov and Vyalyi, showing that any algorithm for solving semidefinite feasibility could be used to compare numbers represented by arithmetic circuits. Note that here we do not have a bound on the degree of the circuits, which could then be exponential in the size of the circuit, and efficiently comparing the (possibly doubly-exponentially large) numbers output by such arithmetic circuits is an old, longstanding problem, which includes the infamous sum-of-square-roots problem as a special case, and which may well not be polynomial-time solvable.}
However, semidefinite feasibility \textit{is} in $\NP(\R) \cap \coNP(\R)$, and one can formulate a sufficiently strong cryptographic conjecture, which would imply the existence of a natural proofs barrier that would apply here.\footnote{In a follow-up to his and Razborov's natural-proofs result \cite{rudich1997super}, Rudich extended the natural proofs barrier as follows. Clearly no pseudorandom generator can fool $\NP$, since in order to distinguish a random from a pseudorandom string, one can always guess the preimage. In his work, Rudich considers the possibility that there exist pseudorandom generators that fool $\coNP$. In other words, he conjectures that there exist pseudorandom generators such that no family of short, efficiently recognizable ($\ZO^\ast$-valued) objects serve to witness that a given string is \textit{not} pseudorandom (not even for a non-negligible fraction of all strings). One could extend Rudich's conjecture from $\coNP$ distinguishers to $\coNP(\R)$ distinguishers: that no family of low-dimensional, efficiently recognizable real-valued objects could serve to witness that a given string is not pseudorandom. Under this generalization of Rudich's conjecture, it necessarily follows that all attempts at approximating complexity using semidefinite feasibility are doomed to fail, since the real-valued dual witnesses could ultimately be used to witness that a given string is not pseudorandom.}. One could argue whether such a strong cryptographic assumption is believable, but such a discussion will soon become irrelevant to our purpose.

Because shortly after we started working on this, Austrin and Risse \cite{austrin2023sum} showed that the sum of squares proof system (SOS) needs degree roughly $S$ to prove, for any given function $f$, that $f$ needs circuits of size $S$. Carefully checking their proof, and doing the necessary adaptations, it also follows from their results that SOS needs degree roughly $2^d$ to prove a depth-$d$ lower-bound on Boolean formulas. And it is possible to formalize the Karchmer--Wigderson theorem in the SOS proof system, and hence it will follow that SOS needs degree roughly $2^d$ to prove a lower-bound of $d$ on the communication complexity of a Karchmer--Wigderson relation. However, a satisfying instance of a semidefinite feasibility problem can be verified in the SOS proof system using a degree-2 proof! It must then follow that, if we define a communication model using our approach, i.e., we generalize communication complexity by formalizing the existence of a deterministic protocol using a HQFP, and relaxing it to a SDFP, then either (1) the proof that our communication model is stronger than the usual deterministic protocols cannot be show by low-degree SOS proofs (``our formalization of communication complexity is weird''), or (2) our generalized communication model can actually solve every single Karchmer--Wigderson game. This follows because our generalized model is such that we always have short, low-degree proofs of any true lower-bound.

The above considerations lead to \textit{no-go theorem}, which (informally stated) says that, unless a weird ``high-degree'' ingredient is introduced somewhere in the formalization (of communication complexity as a HQFP), the model obtained by semidefinite relaxation will be too strong, and will solve all Karchmer--Wigderson relations. We found it remarkable that statements in proof complexity about lengths of proofs imply the existence of algorithms for Karchmer--Wigderson relations, in a large class of computational models!

This no-go theorem should be seen as a natural, expected consequence of the results of Austrin and Risse. But, perhaps owing to our inexperience with proof complexity, it was not easy for us to verify that the formal connection is really there, and so in \Cref{sec:nogo} we provide a formalization and proof of this no-go theorem (\Cref{thm:nogo}).

In light of such a result, one should ask: is it still worthwhile to pursue the stated aim, of formalizing communication complexity using a HQFP, relaxing to a SDFP, and studying the resulting communication model? As it turned out, we went through this formalize-and-relax process twice, and in both times there was something interesting to be found on the other side. In one case we ended up with a communication model which a kind of structured version of the well-known and well-studied $\gamma_2$ norm. In the other case, we ended up with a communication model that has a natural, physical description, and understanding this model revealed to us something non-obvious about the nature of quantum measurements.

And although it is now expected that both models can solve all Karchmer--Wigderson relations, the above no-go theorem is not constructive, and gives us no explicit description of the algorithms in the model that actually do this. So it is still worthwhile to give a constructive proof of this, i.e., to find algorithms in the model for solving Karchmer--Wigderson relations. 

\medskip\noindent
We now describe the two models.

\mysubsubsection{\texorpdfstring{$\gamma_2$}{gamma-2} Communication}

Our first attempt to express a communication protocol as an HQFP proceeds as follows.
A two-party communication protocol computing a function $f: \cX \times \cY \to \{0,1\}$ induces a tree structure over rectangles of $\cX \times \cY$, describing the nodes in the protocol tree, which player speaks at which node, and which child will follow for each possible message sent at each node --- i.e., everything needed to define a protocol, except for the specific combinatorial rectangles which are associated with each node. 
Now, for a given tree structure $\cT$ we will design an HQFP $Q_{\text{protocol}}$ such that solutions to $Q_{\text{protocol}}$ are in 1-1 correspondence with protocols of structure $\cT$ for computing $f$, i.e., associations of rectangles to the nodes of $\cT$ that form a valid protocol for computing $f$. Then, there will exist a protocol with structure $\cT$ computing the function $f$ if and only if there is a solution to $Q_{\text{protocol}}$.

The details are in \Cref{ssec:GammaProgram}, but the central feature of $Q_{\text{protocol}}$ is that we have one variable $A_{t}(x)$ for each node $t$ and each input $x$ of Alice, and one variable $B_{t}(y)$ for each node $t$ and each input $y$ of Bob, so that the product $A_t(x) \cdot B_t(y)$ is to be interpreted as an indicator of whether the input $(x,y)$ belongs to the rectangle associated with node $t$. Given this particular choice of variables, the constraints are the most obvious possible.

In \Cref{ssec:GammaProgram}, we describe the HQFP $Q_{\text{protocol}}$ and relax it into an SDFP $P_{\text{protocol}}$, as discussed above.
It will follow, then, that one can view a solution of $P_{\text{protocol}}$ as a generalization of a protocol computing $f$.
We refer to the solutions of $P_{\text{protocol}}$ as ``$\gamma_2$ protocols'' due to their relationship with the $\gamma_2$ norm.

The $\gamma_2$ norm is a matrix norm, which was introduced to the TCS community by Linial et al.~\cite{LMSS07}. We will give the formal definition in \Cref{sec:sgamma2} in \Cref{ssec:GammaProgram}, but for now it suffices to say the following. If we take any the matrix $M$ which is $1$ inside a combinatorial rectangle, and $0$ outside, i.e. it is the indicator function of a combinatorial rectangle, then its $\gamma_2$ norm is exactly $1$. Indeed, one can define a HQFP $Q_{\text{rectangle}}$ whose feasibility is equivalent to the statement ``the matrix $M$ is the indicator matrix of some combinatorial rectangle'', and then relax it to a SDFP $P_{\text{rectangle}}$, so that the feasibility of $P_{\text{rectangle}}$ is equivalent to the statement ``$\gamma_2(M) \le 1$''. So it is fair to say that matrices with subunit $\gamma_2$ norm are a semidefinite relaxation of the notion of a combinatorial rectangle. Our HQFP $Q_{\text{protocol}}$ is then obtained by putting together HQFPs of the form $Q_{\text{rectangle}}$, one for each node in the protocol structure $\cT$, with some additional constraints to ensure that the rectangles associated with a node and its children form a valid message for that node.

In other words, $Q_{\text{protocol}}$ describes the constraints required so that a collection of rectangles has the structure of a protocol. Analogously, $P_{\text{protocol}}$ will impose a similar structure to a collection of matrices with subunit $\gamma_2$ norm. This is why we call ``$\gamma_2$ protocols'' to solutions to $P_{\text{protocol}}$, and ``$\gamma_2$ communication'' to the resulting communication model.

So, how powerful are $\gamma_2$ protocols? In \Cref{ssec:GammaLB} we prove a discrepancy lower bound for the $\gamma_2$ communication complexity. So, for example, the inner-product mod-2 function cannot be computed by $\gamma_2$ protocols of depth $o(n)$.

On the other hand, in \Cref{ssec:EQGammaProtocol} we design a two-round $\gamma_2$ protocol for the equality function, where in the first round Alice sends 1 of 11 possible messages and in the second round Bob replies with 1 bit. By the usual binary-search reduction of Karchmer--Wigderson relations to equality, it follows that every Karchmer--Wigderson relation can be solved in $\gamma_2$ communication $O(\log n)$.


\mysubsubsection{Quantum Lab Protocols}

Let us begin by contrasting what we will do in \Cref{sec:quantumlab} with what we have done in \Cref{sec:sgamma2}. As before, we will formulate the existence of a two-party deterministic protocol computing $f$ as a HQFP. In the previous section, we had two variables $A_t(x)$ and $B_t(y)$ for every node $t$ in the protocol tree $\cT$, and every input $(x,y) \in X\times Y$. The different starting point here is that our HQFP will have a single variable $C_t(x,y)$. Before, we interpreted $A_t(x) \cdot B_t(y) \in \ZO$ as indicating whether $(x,y)$ is in the rectangle associated with $t$. Now, instead, we let $C_t(x,y) \in \ZO$ indicate the same thing. The constraints of the new program are again designed in the most obvious way possible, so as to ensure that the HQFP is feasible if and only if $f$ can be computed by a deterministic communication protocol with the given structure. We will then relax the quadratic program to a semidefinite program and see what we get.

Notice the difference in approach. In the previous section we had a rationale to obtain the semidefinite program which we obtained: we wanted to add structure to a known rectangle-like notion, the $\gamma_2$ norm, in a similar way to how protocols are obtained from rectangles. The previous model can thus be justified on technical grounds, as, \textit{what happens when we add structure to the $\gamma_2$ norm?}. In contrast, the work in this section began by simply trying to make a different set of constraints where the variables are organized differently. It was surprising to us, then, to discover that the resulting computational model has a natural, functional definition, which can be described as follows.

Alice and Bob work in a idealized quantum laboratory. In this quantum lab, they can prepare any quantum state that they wish, and they can manipulate it without any error using the available equipment. With this lab at their disposal, they play the following ``communication'' game. Before they receive their respective inputs, Alice and Bob are allowed to go to the lab together, and prepare a quantum system in some initial state $\ket{\psi_0}$, known to both. Then they are separated, Alice receives an input $x \in X$, and Bob receives an input $y \in Y$. Their goal is now to compute $f(x,y)$. For this purpose, Alice and Bob take separate turns going to the lab. When one of them is in the lab, she or he is allowed to perform a binary measurement on the quantum system, and write the outcome, $0$ or $1$, in the lab's whiteboard. The measurement that is performed by each player can depend on the input known to her or him, and on the \textit{transcript} of all previous measurement outcomes, which are written in the whiteboard. The question is then: how many times (in the worst case) must Alice and Bob visit the lab, in order to discover $f(x,y)$? Note that, unusually for a quantum model, here we require that Alice and Bob learn $f(x,y)$ without any error. To this minimum number we could call the \textit{(deterministic) quantum-lab complexity of $f$.} 

The first observation is that Alice and Bob can simulate a deterministic protocol. Indeed, if they prepare the two qubit state $\ket{01}$, then Alice can ``communicate'' a $0$ to Bob by measuring the first qubit, which will always be $0$, and she can communicate a $1$ by measuring the second qubit. So this shows, for example, that the two-round quantum-lab complexity of any Boolean function is at most $n+1$, since Alice can communicate their entire input to Bob, and Bob replies with $f(x,y)$. The question is now: can Alice and Bob do better if the lab is quantum? \footnote{As a passing remark, we note that we could have given the very same definition above, but for a \textit{classical laboratory}. In a classical lab, Alice and Bob can prepare any classical state (a distribution over basic states), and measurements correspond orthogonal projections on a fixed basis, followed by renormalization in the $\ell_1$ norm. One can get a sense for the model by imagining a lab made of mechanical contraptions that toss random coins and pull strings and send metal spheres rolling down rails and so on. Every day Alice or Bob go to the lab, and do a ``orthogonal measurement in a fixed basis'', meaning they partition the set of possible outcomes into two, and ask in which of the two sets is the state of the lab. (One can imagine that they look through a window to learn one bit about the state.) As it turns out, this model corresponds to the completely positive relaxation of our HQFP, and it can be shown that, if we require the output to be correct with probability at least $\eps \in [0,1]$, our program gives us exactly the $\eps$-error randomized communication complexity.}

On our part, after discovering this functional description of the model, we were possessed of the following strong intuition: \textit{the measurement that a player is allowed to make depends on her/his input and on the current state $\ket{\psi}$, but if it is a binary measurement, then it cannot reveal more than $1$ bit of information about her/his input, and hence there should exist some kind of information-theoretic lower-bound on the quantum-lab complexity.} We were hoping to prove, at least, that the quantum information complexity \cite{touchette2015quantum} would serve as a lower-bound for quantum-lab complexity.

This intuition, however, turned out to be spectacularly wrong. We were first encouraged by a proof that equality requires $\Omega(n)$ bits to be computed by a two-round quantum-lab protocol (in a two-round protocol Alice does several measurements, then Bob, after which the answer must be known). A simple proof of this, using the quantum pigeonhole principle, appears in \Cref{sec:quantum-lab-2-round-lower-bound}. This early result was encouraging but highly misleading. After a lot of effort trying to prove a lower-bound for 3 rounds, we eventually discovered that equality has a 3-round quantum lab protocol with $O(1)$ complexity. Perhaps this is not surprising, since the information complexity of equality is $O(1)$, and the no-go theorem implies that KW-games will all be easy in the model.

However, a small adjustment to the same protocol revealed that \textit{every Boolean function} can be solved in three rounds with $O(1)$ measurements! This, we did find very surprising, as did everyone to whom we explained the result. On the nature of quantum measurements, we can conclude that although each measurement in the quantum lab can only reveal one bit of information (about $x$ to Bob, and about $y$ to Alice), measurements alone can manipulate the state so that \textit{any} joint bit $f(x,y)$ is revealed.

Perhaps here the reader is tempted to try and solve the puzzle themselves, for which we give the structure of the protocol as a clue: Alice goes to the lab, makes a 1-bit measurement depending on $x$, then Bob goes and makes a two-bit measurement depending on $y$ and on the outcome of Alice's measurement, and then Alice returns to the lab, and does one final 1-bit measurement (depending on $x$ and the previous outcomes) whose answer will be exactly $f(x,y)$. This same protocol structure works for computing any Boolean function $f$, it is only the chosen measurements that vary. Our solution appears in \Cref{sec:quantum-lab-model-collapse}.

\mysubsubsection{Future directions}

We have proposed a specific way of generalizing $\Pi_1$ statements. We would like to suggest a few questions for the future.

\begin{itemize}
    \item What other combinatorial principles can be relaxed by the above approach? An interesting avenue is to investigate the several different combinatorial principles that lie at the basis of TFNP classes, write each of them down by a HQFP, relax to a SDFP, and see what is there. Does this work often? Do we get interesting quantum versions of known principles? In other words, we have an (incomplete) proof system for $\Sigma_1$ and $\Pi_1$ statements, such that every statement or its negation has short proofs. What other interesting theorems can it prove?

    \item Could we take a similar approach using lattice duality? E.g. we could try to express $\Sigma_1$ statements using the closest vector problem (which is NP-hard), and then relax the approximation factor to $\sqrt{n}$, which puts the problem in $\NP \cap \coNP$ \cite{aharonov2005lattice}, and see if the statement is still meaningful.

    \item Could we take a similar approach using stochastic games? Here we have no suggestion for which NP-hard problem could be used, that has stochastic games as a relaxation. 

    \item We have proven that any KW game can be solved by $\gamma_2$ protocols of depth $\le \log(11\times 2) \cdot \log n \approx 4.45 \log n$, i.e. size $\approx n^{4.45}$. However, the best known lower-bounds on formula size are (roughly) cubic \cite{stad1998shrinkage}. Although it seems like a long shot, perhaps one can still prove a super-cubic lower-bound on formula size by constructing an explicit dual to the SDFPs defining $\gamma_2$ protocol for the Karchmer--Wigderson game of some explicit function?

    \item We chose not include the details in this write-up, but it is possible to relax the HQFPs using the completely positive cone, instead of the semidefinite cone. The semidefinite cone is the cone of matrices of inner products of vectors in the entire space, and the completely positive cone is the cone of matrices inner products of vectors in the non-negative orthant. When doing so, one systematically obtains \textit{randomized} versions of the statements, instead of \textit{quantum} versions. We did not explore this much, because completely positive feasibility is still an NP-complete problem. But it might be interesting to see what one gets by such relaxation: maybe new randomized versions of known combinatorial principles?

\end{itemize}

%% file: 01_preliminaries.tex
\newpage
\section{Preliminaries}


\medskip\noindent
We assume that the reader is familiar with Boolean formulas, Boolean circuits, and communication complexity. Recall that the Karchmer--Wigderson theorem states that the minimum depth of a Boolean circuit or formula that computes a given Boolean function $f:\ZO^n\to\ZO$, is equal to the communication complexity of the Karchmer--Wigderson relation $\KW_f$, where Alice is given $x \in f^{-1}(1)$ and Bob is given $y \in f^{-1}(0)$, and they wish to find some $i$ such that $x_i \neq y_i$. A proof can be found in \cite[Section 10.2, see also Chapters 5 \& 10]{KN97}.

\subsubsection*{Discrepancy}
A well-known lower bound for the communication complexity of several models is the discrepancy of a function $f$ (see, e.g., \cite[Section 3.5]{KN97}). Informally speaking, if a function $f$ has a small discrepancy, then any large rectangle $R$ is almost balanced (the number of 1's and 0'z in $R$ is roughly the same).

\begin{definition}\label{def:discrepancy}
    Let $f: \cX \times \cY \to \{0,1\}$ be a function, $R \subseteq \cX \times \cY$ be a rectangle, and $\mu$ be a distribution over $\cX \times \cY$.
    Denote
    \[
    \disc_\mu(R,f) = \Bigl| \Pr_{(x,y) \sim \mu}\bigl[f(x,y) = 0, (x,y) \in R\bigr] -  \Pr_{(x,y) \sim \mu}\bigl[f(x,y) = 1, (x,y) \in R\bigr] \Bigr|.
    \]
    The \emph{discrepancy} of $f$ according to $\mu$ is 
    \[
    \disc_\mu(f) = \max_R \disc_\mu(R,f), 
    \]
    where the maximum is over all rectangles $R \subseteq \cX \times \cY$.
    The \emph{discrepancy} of $f$ is
    \[
    \disc(f) = \min_\mu \disc_\mu(f).
    \]
\end{definition}

\subsubsection*{The notation $\Sigma_1$, $\Pi_1$, $\NP$, $\coNP$, $\NP(\R)$ and $\coNP(\R)$}

We use $\Sigma_1$ and $\Pi_1$ to informally refer to existential and universal statements, respectively. When precision is required, we will use $\NP$ and $\coNP$ for the well-known Boolean complexity classes, and $\NP(\R)$ and $\coNP(\R)$ for the low-degree Blum-Shub-Smale (BSS) variants. The definition is rather technical, but here it is: The BSS model is a variant of the multitape Turing machine where each tape cell holds a real number, and at each step the machine can read the numbers under some of the tape heads, apply a multilinear polynomial to the numbers (which polynomial depends on the state), and write the result back; it can also branch on comparisons between cells, or between a cell and a fixed constant. The low-degree polytime variant imposes the restriction that the computation is syntactically polynomial-degree, meaning that the machine runs in polynomial time, but furthermore: at any given time, for each possible branching that happened before time $t$, the contents of each cell will be a polynomial in the real numbers $x_1, \ldots, x_n$ placed in the tape at the start of the computation, and we then require that the degree of this polynomial to also be $\poly(n)$-bounded (in principle the degree after $t$ steps could be $2^t$ by repeated squaring). Then $\NP(\R)$ is the class of languages $L \subseteq \R^\ast$ for which there exists a low-degree polytime BSS machine $M$ such that $(x_1, \ldots, x_n) \in L \iff \exists (y_1, \ldots, y_m) \in \R^{\poly(n)} M(\bar x, \bar y) = 1$.\footnote{If the reader is wondering why the low-degree restriction, it is because polytime BSS machines without degree constraints can do things that seem too powerful, such as factoring \cite{shamir1979factoring}.}

\subsubsection*{Conic feasibility problems}

Here we discuss duality for conic feasibility problems. 

\begin{definition}
    Let $S, T \subseteq \cH$ denote arbitrary, non-empty subsets of a finite-dimensional real Hilbert space $\cH$. I.e., $\cH = \R^d$ for some $d$, but equipped with a possibly non-standard inner-product $\langle \cdot, \cdot \rangle_\cH$. 

\begin{itemize}
    \item We let $\cl(S)$, the \emph{closure} of $S$, be the set of points $x \in \cH$ for which there exists a sequence $(x_i)_{i \in \N}$ of points in $S$ such that $\|x_i - x\|_\cH \to 0$. We call $S$ \emph{closed} if $S = \cl(S)$.
    
    \item   For $\lambda \in \R$, we denote $\lambda S = \{ \lambda s \mid s \in S \}$, $S + T = \{s + t \mid s \in S, T \in S \}$.

    \item A set $S$ is called \emph{convex} if it contains all the line segments between its points, i.e., $\alpha S + (1-\alpha) S \subseteq S$ for every $0 \le \alpha \le 1$.

    \item $S$ is called a \emph{cone} if $\lambda S \subseteq S$ for all $\lambda \ge 0$. A cone $S$ will be convex iff $S + S \subseteq S$. A cone is called \emph{pointed} if $S \cap -S = \{0 \}$.

     For example, a subspace is a closed convex cone. The non-negative orthant is a closed, convex, pointed cone.

    \item The \emph{polar of $S$}, denoted $S^\ast$, is the set
    \[
    S^\ast = \{ y \in \cH^\ast \mid \forall x \in S\; \langle x, y \rangle_\cH \ge 0 \}.
    \]

\end{itemize}
\end{definition}

\medskip\noindent
\textbf{Examples.} The following sets are closed, convex, pointed cones:
\begin{itemize}
    \item The non-negative orthant $\R^n_{\ge 0}$. It is self-dual, meaning $(\R^n_{\ge 0})^\ast = \R^n_{\ge 0}$.

    \item The set of \textit{positive semidefinite} $n\times n$ matrices $\PSD_n$, which is a subset of the space $\R^{\frac{n(n+1)}{2}}$ of symmetric matrices, with the inner product $\braket{M, N} = \sum_{i,j} M_{i,j} N_{i,j}$.
    
    This set can be alternatively characterized as the set of symmetric matrices with non-negative eigenvalues, or as the set of \textit{Gram} matrices, i.e., matrices equal to $A A^\tr$ for some $n\times m$ matrix $A$, i.e., matrices $M$ of inner products, given by a family of vectors $a_1, \ldots, a_n$ (the rows of $A$), so that $M_{i j} = \braket{a_i \mid a_j}$. It is also self-dual.
    
    \item The set of \textit{completely positive} $n\times n$ matrices $\CP_n \subseteq \R^{\frac{n(n+1)}{2}}$ (also symmetric). This set can be alternatively characterized as the set of symmetric matrices with non-negative eigenvalues whose eigenvectors are entrywise non-negative in the standard basis, or the matrices of the form $M = A A^\tr$ for some $n\times m$ matrix $A$ with non-negative entries, or matrices of inner-products of vectors in the non-negative orthant. Its dual cone is the cone of co-positive matrices, but we will not define it or mention it again.
\end{itemize}

\begin{definition}
    Let $\cK \subseteq \R^n$ be a closed, convex, pointed cone. A conic feasibility problem over $\cK$ is defined by a linear map $\cA:\R^n\to\R^m$ and a point $b \in \R^m$. The problem asks whether there exists an element $Z \in \cK$ such that $\cA(Z) = b$. Such a $Z$ is called a \textit{solution}. If a solution exists, we say that the problem $(\cA,b)$ is \textit{feasible}, or \textit{satisfiable}, and otherwise we say that the problem $(\cA,b)$ is \textit{infeasible}, or unsatisfiable.
\end{definition}

\medskip\noindent
\textbf{Examples.} A linear feasibility problem is a conic feasibility over the non-negative orthant. A semidefinite feasibility problem (SDFP) is a conic feasibility problem over the cone of positive semidefinite matrices.

\subsubsection*{Duality for SDFPs}

The feasibility of a conic feasibility problem over $\cK$ is an existential statement, in fact it it is a $\Sigma_1$ statement provided that $Z \in \cK$ is itself a $\Sigma_1$ statement. A remarkable general fact about conic feasibility is that the \textit{infeasibility} of a conic feasibility problem can \textit{also} be formulated as a $\Sigma_1$ statement . This fact is really non-obvious: it was first proven for SDFPs by Ramana \cite{ramana1997exact} (see \cite{lourencco2023simplified} for a simplified treatment), and for general conic feasibility by \cite{liu2018exact}. This result is an instance of the general phenomenon of \textit{convex duality}, which is also the source of the $\NP \cap \coNP$ inclusions of approximate lattice problems \cite{aharonov2005lattice} and stochastic games (e.g. \cite{polyhedra_equiv_mean_payoff,issac2016jsc}, although here convexity is over the tropical semiring).

The precise statement which is equivalent to the infeasibility of a conic optimization problem, the so called \textit{dual problem}, is not easy to describe in general. It is usually a $\Sigma_1$ statement with another cone as an oracle, usually the polar cone $\cK^\ast$ over a larger dimension, or another related cone.

However, in some cases, a dual problem exists which \textit{is} easy to describe, whose flavor is similar to Farkas' lemma of linear feasibility, and indeed gives exactly Farkas' lemma when applied to the non-negative orthant. It was proven long ago by Ben-Israel:

\begin{theorem}[Ben-Israel \cite{benisrael_linear_1969}]\label{thm:ben-israel} Let $\cK \subseteq \R^n$ be a closed convex cone. Let $\cA: \R^n \to \R^{m}$ be a linear map, and $b \in \R^m$. Suppose that $\ker(\cA) + \cK$ is a closed set (\textit{Ben-Israel's criterion}). Then exactly one of the following two things are true:
\begin{enumerate}[(i)]
  \item\label{itm:alternative1} Either there exists $Z \in \cK$ such that $\cA(Z) = b$,
  \item\label{itm:alternative2} Or there exists $w \in \R^m$ such that $\cA^\tr(w) \in \cK^\ast$ and $\langle w, b \rangle < 0$.
\end{enumerate}
\end{theorem}

A sufficient condition for the closure of $\ker(\cA) + \cK$ is given by the following lemma. It appears in a paper by Berman and Ben-Israel \cite{berman_more_1971}, and there the proof is attributed to A.~Charnes and A.~Lent.

\begin{lemma}[Berman--Ben-Israel criterion]\label{thm:berman-ben-israel} If $L \subseteq \R^n$ is a linear subspace, $S \subseteq \R^n$ is a closed convex cone, and $L \cap S$ is a linear subspace, then $L + S$ is closed. Hence, a sufficient condition for Ben-Israel's criterion to hold is that $\ker(\cA) \cap \cK$ is a linear subspace, for example, $\ker(\cA) \cap \cK = \{0\}$.
\end{lemma}

In all the SDFPs we will consider, we will have the simplest of conditions $\ker(\cA) \cap \cK = \{0\}$.

\subsubsection*{HQFPs, and their relaxation}

A SDFP asks whether there exists a positive semidefinite (symmetric) $n\times n$ matrix $Z$ such that $\cA(Z) = b$, where $\cA$ is a linear map in the entries of $\cZ$ and $b \in \R^m$. In other words, $\cA(Z) = (\braket{A_1, Z}, \ldots, \braket{A_m, Z})$ for some symmetric real matrices $A_1, \ldots, A_m$. Since positive semidefinite matrices are matrices of inner-products, we can rephrase this question as follows: We wish to know whether there exist vectors $a_1, \ldots, a_n \in \R^n$ obeying a set of linear equations in their inner-products $\braket{a_i, a_j}$.

We can now consider the same problem, with the additional constraint that the vectors $a_1, \ldots, a_n$ are scalars (i.e. come from the same $1$-dimensional subspace). This is equivalent to requiring that the solution $Z$ has rank $1$. With this additional constraint, we have a system of linear equations in the quadratic products $a_i \cdot a_j$, and we wish to know whether there exists some choice of scalars that satisfy the system. We call such a problem a \textit{Homogenous Quadratic Feasibility Problem} (HQFP). Naturally, we can take \textit{any} HQFP and relax it to a SDFP by droping the rank-1 restriction, i.e. by replacing scalars with vectors and products with inner-products.

%% file: 02_qphp.tex
\section{The Quantum Pigeonhole Principle}\label{sec:qphp}

The pigeonhole principle (PHP) asserts that placing $p$ pigeons in $h < p$ holes will always result in there being a hole with more than one pigeon. In its simplest form, the quantum pigeonhole principle (QPHP) is the following:

\begin{theorem}[QPHP]\label{thm:qphp}
  Let $\{\lambda\} \cup \{ v_{i,j} \mid i\in [p], j \in [h] \} \subseteq \cH$ be a family of vectors in a Hilbert space $\cH$, such that
  \begin{align*}
    \|\lambda\|^2 & = 1\\
    \sum_{j=1}^h v_{i,j} & = \lambda & \forall i \in [p]\\
    \langle v_{i,j}, v_{i,j'} \rangle & = 0 & \forall i\, \forall j \neq j'\\
  \end{align*}
  I.e., each family $V_i = \{v_{i,j} \mid j \in [h]\}$ decomposes the same unit vector $\lambda$ as a sum of $h$-many orthogonal vectors (\emph{We have $p$ copies of $\lambda$ --- the ``pigeons'' --- and divide each pigeon among $h$ ``holes''}). Suppose that $h < p$. Then, there exists $j \in [h]$ and $i \neq i'$ in $[p]$, such that
  \[
    \langle v_{i,j}, v_{i',j} \rangle \neq 0
  \]
  (\emph{one of the holes must have more than one pigeon}).
\end{theorem}

The above theorem generalizes the Pigeonhole Principle. Indeed, it is equivalent to the Pigeonhole Principle if $\cH$ is one-dimensional. In this case, $\lambda = \pm 1$, and the last two equations imply that, for each $i \in [p]$, $v_{i,j} = \lambda$ for exactly one choice of $j$, and $v_{i,j} = 0$ for the remaining choices. It then follows that there exists a hole $j$ and two pigeons $i \neq i'$ with $v_{i,j} = v_{i,j'} = \pm1$.

We will now see, in \Cref{sec:qphp-conic-program}, how one obtains the QPHP as a semidefinite relaxation of the pigeonhole principle.\footnote{More precisely, as the negation of a semidefinite relaxation of the negation of the pigeonhole principle.} We will also show that the corresponding completely positive relaxation characterizes a probabilistic pigeonhole principle. Then in \Cref{sec:qphp-AM-GM} we prove the QPHP using the AM-GM inequality. This proof is not tight, in the sense that it does not provide the best possible lower-bound on the minimum inner-product $\langle v_{i,j}, v_{i',j} \rangle$ appearing in the theorem. So in \Cref{sec:qphp-geometric}, we prove a tight bound using a geometric argument. Now, we know that the theorem being true implies that there exists an explicit proof in ``canonical form'', namely, a solution to a certain semidefinite feasibility problem. So in \Cref{sec:qphp-dual}, we compute the dual of the semidefinite feasibility problem, and give a solution for it.

\subsection{The Semidefinite Feasibility Problem}\label{sec:qphp-conic-program}

In proof complexity, more specifically in a proof system called Polynomial Calculus, the negation of the pigeonhole principle is sometimes formalized as the following quadratic feasibility problem:
\begin{align}
    \text{There exist } \lambda & \in \R \notag\\
     v_{i,j} & \in \R  & \forall i\in[p], j\in[h]\notag\\
    \text{such that} & &\notag\\
    \lambda^2 & = 1 \notag\\
    \sum_{j=1}^h v_{i,j} & = \lambda & \forall i \in [p] \notag\\
    v_{i,j} \cdot v_{i,j'} & = 0 & \forall i \forall j \neq j'\notag\\
    v_{i,j} \cdot v_{i',j} & = 0 & \forall j \forall i \neq i'\notag
\end{align}
This system is not homogeneous, so it is not immediate how to express it as a SDFP. Nonetheless, we can attempt to naively relax this program to higher dimensions, by replacing real numbers with vectors, and products with inner products. This gives us exactly the negation of the QPHP (\Cref{thm:qphp}):
\begin{align}
    \text{There exists a vector space } & V \notag\\
    \text{ and vectors } \lambda & \in V \notag\\
     v_{i,j} & \in V  & \forall i\in[p], j\in[h]\notag\\
    \text{such that} & &\notag\\
    \|\lambda\|^2 & = 1 \notag\\
    \sum_{j=1}^h v_{i,j} & = \lambda & \forall i \in [p] \label{eq:qphp-sum}\\
    \braket{v_{i,j} \mid v_{i,j'}} & = 0 & \forall i \forall j \neq j'\label{eq:qphp-orth}\\
    \braket{v_{i,j} \mid v_{i',j}} & = 0 & \forall j \forall i \neq i'\label{eq:qphp-php}
    \\
    \intertext{
Again it is not immediate that this is a SDFP, since \ref{eq:qphp-sum} is not directly an equation about inner-products. However, we can replace \ref{eq:qphp-sum} with:
}
    \sum_{j=1}^h \|v_{i,j}\|^2 & = \|\lambda\|^2 & \forall i \in [p] \label{eq:qphp-sum-a}\tag{\ref*{eq:qphp-sum}a}\\
    \sum_{j=1}^h \braket{ v_{i,j} \mid \lambda} & = \|\lambda\|^2 & \forall i \in [p] \label{eq:qphp-sum-b} \tag{\ref*{eq:qphp-sum}b}
\end{align}
To see the equivalence, notice that for each fixed $i \in [p]$, \ref{eq:qphp-orth} states that the $v_{i,j}$ are orthogonal. Under such orthogonality, it is obvious that \ref{eq:qphp-sum} implies \ref{eq:qphp-sum-a} and \ref{eq:qphp-sum-b}, by Pythagoras' Theorem. Conversely, let $\lambda'_i = \sum_j v_{i,j}$. Then \ref{eq:qphp-sum-b} states that $\langle \lambda'_j, \lambda \rangle = \|\lambda\|^2$ and, under orthogonality, Pythagoras' Theorem says $\|\lambda'_i\|^2 = \sum_{j} \|v_{i,j}\|^2$, and so (3) is saying that $\|\lambda'_j\|^2 = \|\lambda\|^2$. These two together imply, by the equality case of Cauchy-Schwarz, that $\lambda'_j = \lambda$.

It is now clear that we have a semidefinite feasibility problem. It can also be seen that taking constraints (1)-(4), and further restricting $\lambda$, $v_{i,j}$ to have dimension $1$, gives us a HQFP, which is equivalent to the negation of the PHP. We will see in \Cref{sec:qphp-dual} that the problem is nice enough that it has a simple dual problem, which is satisfiable if and only if the QPHP is true, and we will also provide an explicit solution to the dual. But for now, we prove the QPHP theorem by other means.


\subsection{A Non-tight Proof Using the AM-GM Inequality}\label{sec:qphp-AM-GM}

We prove a stronger statement that implies a non-tight QPHP, namely, that QPHP holds provided that the number $h$ of holes is sufficiently smaller than the number $p$ of pigeons.

\begin{theorem}[Weak Quantitative QPHP]\label{thm:weak-quantitative-qphp}
    Let $\psi_1, \ldots, \psi_p$ be vectors in a Hilbert space, and for each $i \in [p]$ let $\psi_{i,0}, \psi_{i,1}$ give an orthogonal decomposition of $\psi_i$:
    \[
    \psi_i = \psi_{i,0} + \psi_{i,1} \qquad \psi_{i,0} \bot \psi_{i,1}.
    \]
    Then
    \[
    \left|\sum_{i,j} \braket{\psi_i \mid \psi_j}\right| \le 2 \left(\left|\sum_{i,j} \braket{\psi_{i,0} \mid \psi_{j,0}}\right| + \left|\sum_{i,j} \braket{\psi_{i,1} \mid \psi_{j,1}}\right|\right).
    \]
\end{theorem}

Note two things about the above. First, the theorem does not need to assume orthogonality of the decomposition. Second, this is a pigeonhole principle where we only have two holes. However, by repeated application, we can obtain a weak version \Cref{thm:qphp}, provided there are sufficiently-many more pigeons than holes, as follows. We split the holes into two sets of the same size $\pm 1$, repeatedly, until we are left only with sets containing a single hole. This gives us a (partial) binary tree, and we apply \Cref{thm:weak-quantitative-qphp} repeatedly, starting at the root and then following whichever set of holes that has higher total sum-of-inner-products $\left|\sum_{i,j} \braket{\psi_i \mid \psi_j}\right|$. At the start, the sum-of-inner products is $p^2$, since there are $p$ ``pigeons'', each being the same unit vector. By \Cref{thm:weak-quantitative-qphp}, at the end the sum is at least $\frac{p^2}{4^{\lceil \log h \rceil}} \ge \frac{p^2}{4 h^2}$. Since the decomposition is orthogonal, the norms of $\psi_{i,j}$ cannot increase, and so the contributions of the squared norms $\braket{\psi_{i,j} \mid \psi_{i,j}}$ sum to at most $p$. Hence, if the total sum of all inner products is greater than $p$, two distinct pigeons must have non-zero inner-product. This will happen whenever $h < \frac{1}{4} \sqrt{p}$. So we cannot place $p$ pigeons into fewer than $\frac{1}{4} \sqrt{p}$ holes without two pigeons overlapping. Of course, this is not optimal. But \Cref{thm:weak-quantitative-qphp} has a short proof that is easy to check. This proof was suggested to us by Carlos Florentino, who approached \Cref{thm:qphp} as a fun puzzle in linear algebra.


\begin{proof}
    Let $A$ be the matrix whose rows are $\psi_1, \ldots, \psi_p$. Then
    \[
    \left|\sum_{i,j} \braket{\psi_i \mid \psi_j}\right| = |A \cdot A^\tr|.
    \]
    Let $A_b$, $b\in\ZO$, be the matrix whose rows are $\psi_{1,b}, \ldots, \psi_{p,b}$. Then
    \[
    |A \cdot A^\tr| = |(A_0 + A_1) \cdot (A_0^\tr + A_1^\tr)| \le |A_0 \cdot A_0^\tr| + |A_0 \cdot A_1^\tr| + |A_1 \cdot A_0^\tr| + |A_1 \cdot A_1^\tr|,
    \]
    and the theorem follows from the following \textit{AM/GM inequality for matrices}. For any two matrices $B, C$ of compatible dimension:
    \[
    |B \cdot C^\tr| \le \frac{|B \cdot B^\tr| + |C \cdot C^\tr|}{2}
    \]
    (note that the absolute value is only needed if the two matrices being multiplied are not the same). The proof of this is a direct calculation, using the AM/GM inequality for reals. Let $\beta(x)$, $\gamma(y)$ index the rows of $B$ and $C$, respectively. Then:
    \begin{align*}
        |B \cdot C^\tr| & = \left|\sum_{x y} \braket{\beta(x) \mid \gamma(y)}\right|\\
        & = \left|\sum_i \left(\sum_x \beta(x)_i\right) \left(\sum_y \gamma(y)_i\right)\right|\\
        & \le \sum_i \frac{(\sum_x \beta(x)_i)^2 + (\sum_y \gamma(y)_i)^2}{2}    \tag*{(AM/GM)}\\
        & = \frac{1}{2}\sum_{x x'} \braket{\beta(x) \mid \beta(x')} + \frac{1}{2}\sum_{y y'} \braket{\gamma(y) \mid \gamma(y')}\\
        & = \frac{|B \cdot B^\tr| + |C \cdot C^\tr|}{2}. \tag*{\qedhere}
    \end{align*}
\end{proof}
 
\subsection{A Tight Proof Using a Geometric Argument}\label{sec:qphp-geometric}

In Theorem \ref{thm:qphp} we consider the vectors $\{\lambda\} \cup \{ v_{i,j} \mid i\in [p], j \in [h] \} \subseteq \cH$ such that the vectors $v_{i,j}$ form an orthogonal decomposition of the unit vector $\lambda$. The theorem then claims that there must be $j \in [h]$ and $i \neq i' \in [p]$ such that $\ip{v_{i,j}}{v_{i',j}} \neq 0$. Here $\lambda$ represents the initial state of each pigeon and $v_{i,j}$ the part of pigeon $i$ in hole $j$. In this section we will consider a more general case where the initial states can be different. That is, we have initial states $\{v_i\}_{i \in [p]}$ which are all unit vectors in $\cH$. The vectors $\{v_{i,j}\}_{j \in [h]}$ are an orthogonal decomposition of $v_i$. What we will show is the following.

\begin{theorem}[Quantitative QPHP]\label{thm:qqphp}
  Let $\{v_i \mid i \in [p]\} \cup \{ v_{i,j} \mid i\in [p], j \in [h] \} \subseteq \cH$ be a family of vectors in a finite-dimensional Hilbert space $\cH$, such that
  \begin{align*}
    \|v_i\|^2 & = 1 & \forall i \in [p]\\
    \sum_{j=1}^h v_{i,j} & = v_i & \forall i \in [p]\\
    \langle v_{i,j}, v_{i,j'} \rangle & = 0 & \forall i\, \forall j \neq j'\\
  \end{align*}
  I.e., each family $V_i = \{v_{i,j} \mid j \in [h]\}$ decomposes $v_i$ as a sum of $h$-many orthogonal vectors. Let 
  \[
  \beta = \frac{1}{p(p-1)} \sum_{i \neq i'} \ip{v_i}{v_{i'}}
  \]
  (the average overlap between the initial states of the pigeons).
  Then, there exists $j \in [h]$ and $i \neq i'$ in $[p]$, such that
  \[
    \ip{v_{i,j}}{v_{i',j}} \geq \frac{1}{h^2} \left( \beta - \frac{h-1}{p-1} \right).
  \]
  Furthermore, for all choices of $\beta \ge 0, p \ge h \ge 1$, this is the best possible lower-bound holding for all such families of vectors.
\end{theorem}

\begin{proof}
    Our first step is a symmetrization. We will consider a new system of vectors in $\cH^{\oplus p!h!}$ defined as follows.
\begin{align*}
    w_i &:= \frac{1}{\sqrt{p!h!}} \cdot \operatornamewithlimits{\oplus}_{\sigma \in S_p,\tau \in S_h} v_{\sigma(i)}\\
    w_{i,j} &:= \frac{1}{\sqrt{p!h!}}  \cdot \operatornamewithlimits{\oplus}_{\sigma \in S_p,\tau \in S_h} v_{\sigma(i),\tau(j)}
\end{align*}

Note that $w_i$ is still of unit norm. Furthermore for each $\sigma \in S_p,\tau \in S_h$ and $i \in [p]$ the vectors $\{v_{\sigma(i),\tau(j)}\}_{j \in [h]}$ are still an orthogonal decomposition of $v_{\sigma(i)}$. Hence $\{w_{i,j}\}_{j \in [h]}$ continues to be an orthogonal decomposition of $w_i$.

These symmetrized vectors are very useful to us since (a) they have much more structure to work with and (b) the worst-case overlap between pigeons in a hole for the symmetrized pigeons is at most the worst-case overlap for the unsymmetrized pigeons. We elaborate on this in the following analysis of some important inner products of our symmetrized pigeons.
\begin{itemize}
    \item $\ip{w_i}{w_i} = 1$ for all $i$.
    \item $\ip{w_i}{w_{i'}} = \frac{1}{p!}\sum_{\sigma} \ip{v_{\sigma(i)}}{v_{\sigma(i')}}$, where the right hand side is the same expression for all $i \neq i'$.\\
    Note that this is exactly the value $\beta$.
    
    \item $\ip{w_{i,j}}{w_{i,j}} = \frac{1}{p!h!}\sum_{\sigma,\tau} \ip{v_{\sigma(i),\tau(j)}}{v_{\sigma(i),\tau(j)}}$ which is the same for all $i,j$.\\
    Since $\sum_j \ip{w_{i,j}}{w_{i,j}} = \ip{w_i}{w_i}$, this must equal $1/h$ for every $i,j$.
    \item $\ip{w_{i,j}}{w_{i',j}} = \frac{1}{p!h!}\sum_{\sigma,\tau} \ip{v_{\sigma(i),\tau(j)}}{v_{\sigma(i'),\tau(j)}}$ which is the same for all $i \neq i',j$.\\
    Note that this value is the overlap between any two pigeons in any hole for the symmetrized pigeons. It is clearly at most $\max_{i \neq i' \in [p],j \in [h]} \ip{v_{i,j}}{v_{i'j}}$, which is the worst-case overlap for the unsymmetrized pigeons.\\Since this is an important value, we will call this value $\alpha$.
    \item[] The following two inner products have no innate significance, but are used in the proof.
    \item $\ip{w_i}{w_{i,j}} = \frac{1}{p!h!}\sum_{\sigma,\tau} \ip{v_{\sigma(i)}}{v_{\sigma(i),\tau(j)}}$ which is the same for all $i,j$.
    \item $\ip{w_i}{w_{i',j}} = \frac{1}{p!h!}\sum_{\sigma,\tau} \ip{v_{\sigma(i)}}{v_{\sigma(i'),\tau(j)}}$ which is the same for all $i \neq i',j$.
\end{itemize}

Now we only need to prove that the value $\alpha = \ip{w_{i,j}}{w_{i',j}}$ must be at least $\frac{1}{h^2}\left(\beta - \frac{h-1}{p-1}\right)$. This proof will involve analyzing families of vectors having equal length and having the same overlap between them. We call such a family of vectors a ``flower'', and we will need the following properties.

\begin{claim}\label{clm:qphp-flower}
Let $r_1,\dots,r_d$ be vectors in a Hilbert space such that $\|r_i\|^2 = a$ for all $i$ and $\ip{r_i}{r_{i'}} = b$ for all $i \neq i'$. Then
\begin{enumerate}
    \item $b \geq -\frac{a}{d-1}$.
    \item Any $a \geq b$ satisfying the above is achievable.
    \item $\sum r_i = 0$ if and only if $b = -\frac{a}{d-1}$.
\end{enumerate}
\end{claim}

\begin{proof}
    The lower bound on $b$ can be easily seen using the fact that Gram matrices are the same as PSD matrices. The Gram matrix $M$ of the vectors $r_i$ is a $d \times d$ matrix with the diagonal entries being $a$ and the others being $b$. Letting $u$ denote the all-$1$ vector, $u^{\tr} M u = d(a + (d-1)b)$. Since this must be at least $0$, we have $b \geq -a/(d-1)$. 
    
    The second part can also be seen using the connection to PSD matrices. Let $M$ be the $d \times d$ matrix with $a$ on the diagonals and $b$ elsewhere. $M$ has $u$ as an eigenvector with eigenvalue $d(a + (d-1)b)$. Furthermore for each $i \in \{2,\dots,d\}$ the vector $e_1 - e_i$ is an eigenvector with eigenvalue $a-b$. These $d$ eigenvectors are independent, and so this shows that $M$ is PSD, and hence a Gram matrix of some vectors.

    For the third part, if $\sum r_i = 0$ then $\ip{r_1}{\sum r_i} = 0$. But $\ip{r_1}{\sum r_i} = a + b(d-1)$, so this implies $b = -a/(d-1)$. Conversely if $b = -a/(d-1)$ then for all $i$, $\ip{r_i}{\sum r_j} = a + b(d-1) = 0$ and so $\sum r_j$ must be orthogonal to each $r_i$. Hence $\sum r_j \perp \lspan(\{r_i\}_{i \in [d]})$ and so $\sum r_i = 0$.
\end{proof}

Now back to our proof. Let $W = \lspan(\{w_i\}_{i \in [p]})$. Fix a pigeon $i$. Note that the vector $\{\ip{w_{i,j}}{w_{i'}}\}_{i' \in [p]} \in \R^p$ is the same for all $j \in [h]$. Hence the projection to $W$, $\Pi_W w_{i,j}$, is the same vector for all $j$. But since $\sum_j w_{i,j} = w_i$, we know $w_{i,j} = w_i/h + x_{i,j}$ where $x_{i,j} \perp W$. 
And since $\ip{w_i}{w_i} = 1$ and $\ip{w_{i,j}}{w_{i,j}} = 1/h$, we know $\ip{x_{i,j}}{x_{i,j}} = 1/h - 1/h^2$

Now we consider two pigeons $i,i'$ in a hole $j$. We can expand $\ip{w_{i,j}}{w_{i',j}} = \ip{w_i/h + x_{i,j}}{w_{i'}/h + x_{i',j}} = \ip{w_i}{w_i'}/h^2 + \ip{x_{i,j}}{x_{i',j}}$. Hence $\ip{x_{i,j}}{x_{i',j}} = \alpha - \beta/h^2$.

This tells us that the vectors $\{x_{i,j}\}_{i \in [p]}$ form a flower. We can use Claim~\ref{clm:qphp-flower} with $d = p$, $a = 1/h-1/h^2$ and $b = \alpha - \beta/h^2$. Hence
\begin{align*}
    \alpha - \frac{\beta}{h^2} &\geq -\left(\frac{1}{h}-\frac{1}{h^2}\right)/(p-1)\\
    \implies \alpha &\geq \frac{1}{h^2}\left(\beta - \frac{h-1}{p-1}\right)
\end{align*}
which is what we set out to prove.

To prove the tightness of this result, we need to exhibit a tight example of distributing pigeons among pigeonholes. Let $\alpha$ denote the worst-case overlap of two pigeons in a hole. We want to exhibit an example where $\alpha = \frac{1}{h^2}\left(\beta - \frac{h-1}{p-1}\right)$. We follow the path set for us by the symmetrization.

We choose three sets of vectors:
\begin{itemize}
    \item $\{v_i\}_{i \in [p]}$ is a flower with $d = p, a = 1, b = \beta$. Such a flower ought to exist because if a setting of inital pigeons is possible with average overlap $\beta$, then their symmetrization will result in such a flower.
    \item $\{s_j\}_{j \in [h]}$ is a flower with $d = h, a = 1/h-1/h^2, b = -(1/h-1/h^2)/(h-1) = -1/h^2$. Such a flower exists by \Cref{clm:qphp-flower}.
    \item $\{t_i\}_{i \in [p]}$ is a flower with $d = p, a = 1/h-1/h^2, b = \alpha-\beta/h^2$. Such a flower exists since it can be seen that $b = -a/(p-1)$, and by \Cref{clm:qphp-flower}.
\end{itemize}

We now consider the initial pigeons $\{v_i\}_{i \in [p]}$ along with decompositions \[ v_{i,j} = \frac{v_i}{h} \oplus \frac{t_i \otimes s_j}{\sqrt{1/h - 1/h^2}}. \]

It is easy to verify that
\begin{itemize}
    \item $\braket{v_i,v_i} = 1$,
    \item $\braket{v_i,v_{i'}} = \beta$,
    \item $\sum_j v_{i,j} = v_i \oplus \frac{t_i}{\sqrt{1/h - 1/h^2}}\otimes (\sum_j s_j) = v_i$ (by \Cref{clm:qphp-flower}),
    \item $\braket{v_{i,j},v_{i,j'}} = 1/h^2 + (1/h-1/h^2)(-1/h^2)/(1/h-1/h^2) = 0$, and
    \item $\braket{v_{i,j},v_{i',j}} = \beta/h^2 + (\alpha - \beta/h^2)(1/h-1/h^2)/(1/h-1/h^2) = \alpha$.
\end{itemize}
\end{proof}

\subsection{An Explicit Proof via Duality}\label{sec:qphp-dual}

In \Cref{sec:qphp-conic-program}, we displayed a semidefinite feasibility problem equivalent to the negation of the QPHP. It is not hard to see that this feasibility problem obeys the criterion of Berman and Ben-Israel (\Cref{thm:berman-ben-israel}), since setting all constants of the equations equal to $0$, the initial vector $\lambda$ is $0$, and since all the other vectors are orthogonal decompositions of $\lambda$, the only possible solution is when the vectors are all $0$. And so it has a simple dual as in \Cref{thm:ben-israel}, which is computed so that the QPHP is true if and only if there exists  $W \in \PSD_{1 + p h}$ of the form:
 \begin{align*}
    \begin{pNiceArray}{c|c|ccc|c}
        y^{(0)} - 2\sum_i y^{(1a)}_i - \sum_i y^{(1b)}_i & \cdots & \cdots & y^{(1a)}_i & \cdots & \cdots \\\hline
        & \ddots &\mathsf{diag}(\dots & y_{i, 1, j}^{(2)} & \dots) & \dots\\\hline
        & & \ddots && y^{(3)}_{j i i'}&\\
        & & & y^{(1b)}_i &&\\
        && y^{(3)}_{j i i'}&&\ddots&\\\hline
        & & &&& \ddots
    \end{pNiceArray}
\end{align*}
with $y_0 < 0$. (The numbers inside the superscript parenthesis correspond to the equations in the primal.) Since such an explicit proof exists, one should try to find it. One such dual solution is:
\[
  W = \begin{pNiceMatrix} 
    1 &-\frac{1}{h} & \dots & -\frac{1}{h} & \dots & -\frac{1}{h} & \dots & -\frac{1}{h}\\
    & \frac{1}{h} & \Block{2-2}{ \frac{1}{h}} & &&  &  & \\
    &  & \Ddots&                              &&     \Block{2-2}{ }     &  & \\
    &  & & \frac{1}{h}                            &&               &  & \\
    &         &                             & &&  \frac{1}{h} & \Block{2-2}{ \frac{1}{h}} & \\
    &  &       &                              &&             & \Ddots & \\
    &  & &                                    &&               &  & \frac{1}{h}\\
  \end{pNiceMatrix},
\]
I.e., we set $y^{(0)} =1 - \frac{p}{h}$, $y^{(1a)}_i = -\frac{1}{h}$, $y^{(1b)}_i = y^{(3)}_{j i i'} = \frac{1}{h}$, and $y^{(2)}_{i j j'} = 0$. Hence $y^{(0)} < 0$ precisely when $p > h$. This dual solution is PSD, of rank 1! Since the variables $y^{(2)}_{i j j'}$ are equal to $0$, we have also proved that QPHP follows from constraints (1a), (1b) and (3), without needing the equation (2), which states that all parts of each pigeon are orthogonal.\footnote{Note that we used this orthogonality to apply the Berman and Ben-Israel criterion, without which we have no strong duality for the program given above. But weak duality is enough to conclude that the negation of the QPHP is false.}

%% file: 03_sgamma2.tex
\section{\texorpdfstring{$\gamma_2$}{gamma-2} Communication}\label{sec:sgamma2}

In this section, we introduce a generalization of deterministic protocols.
We call these generalized protocols ``$\gamma_2$ protocols'' because of a connection with the $\gamma_2$ norm of matrices.
The $\gamma_2$ norm was introduced to the TCS community by Linial et al.~\cite{LMSS07} to study sign matrices.
\begin{definition}
    Let $A \in \R^{m \times n}$ be a matrix. Then,
    \[
    \gamma_2(A) = \min \{ r(X) r(Y) \mid A = X Y^\tr \},
    \]
    where $r(M)$ is the largest $\ell_2$ norm of a row of the matrix $M$.
\end{definition}

One can see a matrix $A$ with $\gamma_2(A) \leq 1$ as a generalization of a rectangle.
Let $\cX$ and $\cY$ be sets and $R = A \times B$ be a rectangle, where $A \subseteq \cX$ and $B \subseteq Y$.
Let $M_R = \{0,1\}^{\cX \times \cY}$ be a matrix representing the rectangle $R$, i.e., $M_R[x,y] = 1$ if and only if $(x,y) \in R$.
We can decompose the matrix $M_R$ as $M_R = u v^\tr$, where $u \in \{0,1\}^\cX$ and $v \in \{0,1\}^\cY$ are the characteristic vectors of the sets $A$ and $B$, respectively.
Clearly, $r(u) = r(v) = 1$, if we take the vectors $u$ and $v$ as matrices with one column.
Thus, $\gamma_2(M_R) \leq 1$. From the Cauchy-Schwarz inequality, it follows that $\gamma_2(M_R) = 1$. 
Hence, one can think of matrices with $\gamma_2(M) \le 1$ as a generalization of the notion of a combinatorial rectangle.\footnote{In fact, it is possible to write down a HQFP whose solutions are precisely indicator matrices of combinatorial rectangles, and whose semdefinite relaxations are precisely matrices of subunit $\gamma_2$ norm. We leave this as an exercise, which should be very doable after reading \Cref{ssec:GammaProgram}.} This line of thought bore many fruits in the study of communication complexity, such as lower bounds, lifting theorems, the ability to approximate PP-communication-complexity using semidefinite programming, etc, see \cite{lee2009lower} for a survey.

However, a protocol is more than just a rectangle, it is a \textit{structured collection} of rectangles. One can then naturally wonder if we can extend this analogy to include protocols, meaning, we wish to have structured collections of matrices with subunit $\gamma_2$ norm, in a similar way to how protocols are structured collections of rectangles.

Our HQFP to SDFP approach gives us a natural way of doing this, which results in a computational model, which we call \textit{$\gamma_2$ communication.} 
In \Cref{ssec:GammaProgram}, we define a HQFP $Q_{f,\cT}$ whose solutions are exactly deterministic protocols with a certain structure $\cT$ for computing $f$, and we relax it into a SDFP $P_{f,\cT}$ whose solutions will then be deterministic $\gamma_2$ protocols with structure $\cT$ for computing $f$. In \Cref{ssec:GammaLB}, we show how the $\gamma_2$ protocols add structure to a collection of matrices with subunit $\gamma_2$ norm, and we prove a lower bound for $\gamma_2$ protocols using discrepancy.
In \Cref{ssec:EQGammaProtocol}, we present a $\gamma_2$ protocol of depth $O(1)$ for computing the equality function, which implies the existence of an $O(\log n)$ depth protocol for solving any Karchmer--Wigderson game.

\subsection{Definition of the Model}
\label{ssec:GammaProgram}

We now describe a family of algorithms that generalize communication protocols. 
For this purpose, we start by giving a definition of deterministic communication protocols. 
This definition is idiosyncratic, in that it is given by way of a quadratic feasibility problem. 
It will not be immediately obvious why the constraints are chosen the way they are, but it will be possible to see that this feasibility problem is completely equivalent to the usual definition of deterministic protocols.
We will then take that same quadratic feasibility problem, and relax it into a conic feasibility problem, where quadratic products are replaced with inner products. 
This will give us the definition of \emph{$\gamma_2$ protocols}.

A \emph{(two-player, binary) protocol structure} is a finite binary rooted ordered tree $\cT$. 
Each internal node $t \in \cT$ is either an \emph{Alice's node} or a \emph{Bob's node} (but not both). 
We will denote the root of $\cT$ by $\lambda$ (the empty binary string), and the two children of an internal node $t \in \cT$ are denoted by the binary strings $t0$ and $t1$ so that any node is denoted by the binary string which goes from the root to it. 

We then define a \emph{(two-player, binary) deterministic protocol} as a tuple $\pi = (\cX\times\cY, \cT, A, B)$, where $\cX\times\cY$ is a finite product set of \emph{inputs}, $\cT$ is a protocol structure, and $A$ and $B$ are a collection of maps $A_t:\cX\to\R$ and $B_t:\cY\to\R$, for each node $t$ of $\cT$, satisfying the following restrictions.

\begin{description}


\item[Root constraints.] For the root $\lambda$ of $\cT$ we will have the following constraints:
\begin{align*}
  A_\lambda(x) \cdot A_\lambda(x') & = 1 & \forall x,x' \in \cX\\
  B_\lambda(y)\cdot B_\lambda(y') & = 1 & \forall y,y' \in \cY\\
  A_\lambda(x) \cdot B_\lambda(y) & = 1 & \forall (x,y) \in \cX\times\cY
\end{align*}
These imply that every $A_\lambda(x)$ and $B_\lambda(y)$ are either all $1$, or all $-1$.

\item[Alice's nodes constraints.] Let $t \in \cT$ be an Alice's node with two children $t0, t1$. Think that Alice sends a bit $i$ to Bob when going into $t i$. We impose the following constraints.
\begin{align*}
  A_{t0}(x)^2 + A_{t1}(x)^2 & = A_{t}(x)^2 & \forall x \in \cX\\
  A_{t0}(x)\cdot A_t(x) + A_{t1}(x)\cdot A_t(x) & = A_{t}(x)^2 & \forall x \in \cX\\
  A_{t0}(x)\cdot A_{t1}(x) & = 0 & \forall x \in \cX\\
  B_{t0}(y)^2  & = B_{t}(y)^2 & \forall y \in \cY\\
  B_{t1}(y)^2  & = B_{t}(y)^2 & \forall y \in \cY\\
  B_{t0}(y)\cdot B_t(y) & = B_{t}(y)^2 & \forall y \in \cY\\
  B_{t1}(y)\cdot B_t(y) & = B_{t}(y)^2 & \forall y \in \cY
\end{align*}

Take these constraints together. 
By hypothesis, we assume that $A_{t}(x), B_{t}(y) \in \{0,\pm 1\}$, moreover the signs of every non-zero $A_{t}(x)$ and $B_{t}(y)$ are the same.
We conclude (from the last 4 constraints) that $B_{t0}(y) = B_{t1}(y) = B_{t}(y)$ for every $y$, and (from the first three constraints) that for each $x$ we must choose either $A_{t0}(x) = A_{t}(x)$ and $A_{t1}(x) = 0$, or $A_{t0}(x) = 0$ and $A_{t1}(x) = A_{t}(x)$. Thus, if we think of $A_{t'}$ and $B_{t'}$ as subsets of $\cX$ and $\cY$, respectively, these constraints mean that $A_{t0}$ and $A_{t1}$ form a partition of $A$, whereas $B$ is not modified. That is the usual definition of a protocol.

\item[Bob's nodes constraints.] The constraints for Bob's nodes are analogous to Alice's node constraints.
\end{description}

Let $f \subseteq \cX\times\cY\times\cZ$ be a relation with an output set $\cZ \subseteq \ZO^k$ and let $\pi = (\cX\times\cY,\cT,A,B)$ be a deterministic communication protocol. 
We say that \emph{$\pi$ computes $f$} if the depth of every leaf $\ell \in \cT$ is at least $k$, and the collections $A$ and $B$ satisfy the following constraints. 
\begin{description}
    \item[Computational constraints.]  For every leaf $\ell \in \cT$ of the form $\ell = t z$ for some $z \in \ZO^k$ we have the following constraints:
    \begin{align*}
    A_\ell(x) \cdot B_\ell(y) & = 0 & \forall (x,y) \in \cX \times \cY \text{ s.t. } (x,y,z) \notin f
\end{align*}
\end{description}
I.e., we consider the last $k$ bits of the protocol as the output.
In a standard language of protocols, the computational constraints assert the following.
Consider a leaf $\ell$ of $\cT$ that outputs $z \in \cZ$.
Let $R_\ell = C \times D$ be the rectangle associated with the leaf $\ell$ and let $(x,y,z) \not \in f$.
Then, it holds that $(x,y) \not \in R_\ell$.
If we think of $A_\ell$ and $B_\ell$ as characteristic functions of $C$ and $D$, then $A_\ell(x) \cdot B_\ell(y) = 0$ implies that $x \not \in C$ or $y \not \in D$.
That means $(x,y) \not \in R_\ell$ indeed.

The \emph{deterministic communication complexity} of $f$, denoted $\DCC(f)$, is the smallest depth of a protocol structure $\cT$ such that there exists a deterministic communication protocol $\pi = (\cX\times\cY,\cT,A,B)$ that computes $f$.

From the above, it follows that for every fixed protocol structure $\cT$, the predicate ``$f$ can be computed by a deterministic protocol with the protocol structure $\cT$'' can be written as a quadratic feasibility problem. 
As discussed in the introduction, we relax the quadratic feasibility problem into a positive semidefinite feasibility problem.

A \emph{(binary, two-player) $\gamma_2$ deterministic protocol} is a tuple $\pi = (\cX\times\cY, \cT, d, \alpha, \beta)$, where $\cX\times\cY$ is a finite product set of \emph{inputs}, $\cT$ is a protocol structure, and $\alpha$ and $\beta$ are collections of maps $\alpha_t:\cX\to\R^d$ and $\beta_t:\cY\to\R^d$, for each node $t \in \cT$, satisfying a number of constraints below -- that arise from relaxation of the standard protocol constraints described above, where we replace the multiplication by the standard inner product $\ip{\cdot}{\cdot}$ in $\R^d$.

\begin{description}


\item[Root constraints.] For the root $\lambda$ of $\cT$ we have the following constraints.
\begin{align*}
  \langle\alpha_\lambda(x),\alpha_\lambda(x')\rangle & = 1 & \forall x,x' \in \cX\\
  \langle\beta_\lambda(y),\beta_\lambda(y')\rangle & = 1 & \forall y,y' \in \cY\\
  \langle\alpha_\lambda(x),\beta_\lambda(y)\rangle & = 1 & \forall (x,y) \in \cX\times\cY
\end{align*}
This implies that every $\alpha_\lambda(x)$ and $\beta_\lambda(x)$ is the same unit-length vector (in $\ell_2$ norm).

\item[Alice's nodes constraints.] Let $t \in \cT$ be an Alice's node with children $t0, t1$. We impose the following constraints.
\begin{align*}
  \|\alpha_{t0}(x)\|^2 + \|\alpha_{t1}(x)\|^2 & = \|\alpha_{t}(x)\|^2 & \forall x \in \cX\\
  \langle\alpha_{t0}(x),\alpha_{t}(x)\rangle + \langle\alpha_{t1}(x),\alpha_{t}(x)\rangle & = \|\alpha_{t}(x)\|^2 & \forall x \in \cX\\
  \langle\alpha_{t0}(x),\alpha_{t1}(x)\rangle & = 0 & \forall x \in \cX\\
  \|\beta_{t0}(y)\|^2  & = \|\beta_{t}(y)\|^2 & \forall y \in \cY\\
  \|\beta_{t1}(y)\|^2  & = \|\beta_{t}(y)\|^2 & \forall y \in \cY\\
  \langle\beta_{t0}(y),\beta_{t}(y)\rangle & = \|\beta_{t}(y)\|^2 & \forall y \in \cY\\
  \langle\beta_{t1}(y),\beta_{t}(y)\rangle & = \|\beta_{t}(y)\|^2 & \forall y \in \cY
\end{align*}

The above constraints together are equivalent to saying (using the Cauchy-Schwarz inequality and the Pythagorean theorem) that for any $x \in \cX$ and $y \in \cY$, we have that $\alpha_{t}(x) = \alpha_{t0}(x) + \alpha_{t1}(x)$, $\alpha_{t0}(x)$ and $\alpha_{t1}(x)$ are orthogonal, and $\beta_{t}(y) = \beta_{t0}(y) = \beta_{t1}(y)$.

\item[Bob's nodes constraints.] The constraints for Bob's nodes are analogous to Alice's node constraints.
\end{description}

Let $f \subseteq \cX\times\cY\times\cZ$ be a relation with output set $\cZ \subseteq \ZO^k$ and let $\pi = (\cX\times\cY, \cT, d, \alpha, \beta)$ be a $\gamma_2$ protocol. 
We say that \emph{$\pi$ computes $f$} if the depth of every leaf $\ell \in \cT$ is at least $k$, and the collections $\alpha$ and $\beta$ satisfy the following constraints. 
\begin{description}
    \item[Computational constraints.]  For every leaf $\ell \in \cT$ of the form $\ell = t z$ for some $z \in \ZO^k$ we have the following constraints:
    \begin{align*}
    \langle \alpha_\ell(x), \beta_\ell(y) \rangle & = 0 & \forall (x,y) \in \cX \times \cY \text{ s.t. } (x,y,z) \notin f
\end{align*}
\end{description}

The \emph{deterministic $\gamma_2$ communication complexity} of $f$, $\GDCC(f)$, is the smallest depth of a protocol structure $\cT$ such that there exists a $\gamma_2$ deterministic protocol $\pi = (\cX\times\cY,\cT,d,\alpha,\beta)$ that computes $f$.

\subsection{A Lower-bound Using Discrepancy}
\label{ssec:GammaLB}

As we discussed above, protocols induce a tree-like structure over rectangles.
We will show an analogous property of $\gamma_2$ protocols. 
Formally, let $\pi$ be a protocol computing a relation $f \subseteq \cX \times \cY \times \cZ$ with a protocol structure $\cT$.
For a node $t$ of $\cT$, there is a rectangle $R^\pi_t = A^\pi_t \times B^\pi_t \subseteq \cX \times \cY$ containing all input pairs for which the protocol $\pi$ follow the path from the root $\lambda$ of $\cT$ to the node $t$.
For the rectangles $R_t$'s we have the following.
\begin{enumerate}
\item For the root $\lambda$ of $\cT$, it holds that $R^\pi_\lambda = \cX \times \cY$.
\item For a node $t$ of $\cT$ with two children $t0$ and $t1$, it holds that $R^\pi_t = R^\pi_{t0} \dot\cup R^\pi_{t1}$. Moreover, if $t$ is an Alice's node, then $A^\pi_t = A^\pi_{t0} \dot\cup A^\pi_{t1}$ and $B^\pi_t = B^\pi_{t1} = B^\pi_{t0}$. Analogously, if $t$ is a Bob's node, then $B^\pi_t = B^\pi_{t0} \dot\cup B^\pi_{t1}$ and $A^\pi_t = A^\pi_{t0} = A^\pi_{t1}$.
\item For a leaf $\ell$ of $\cT$ outputting $z \in \cZ$, it holds that for each $(x,y) \in \cX \times \cY$ with $(x,y) \in R^\pi_\ell$ we have $(x,y,z) \in f$.
\end{enumerate}

For a $\gamma_2$ protocol $\pi = (\cX \times \cY, \cT, d, \alpha,\beta)$, and a node $t$ of $\cT$ we define a matrix $M^\pi_t \subseteq \R^{\cX \times \cY}$ as $M^\pi_t[x,y] = \ip{\alpha_t(x)}{\beta_t(y)}$.
The next theorem shows that the matrices $M_t$'s have analogous properties to rectangles of protocols.

\begin{theorem}
\label{thm:GammaStructure}
Let  $\pi = (\cX \times \cY, \cT, d, \alpha,\beta)$ be a $\gamma_2$ protocol computing a relation $f \subseteq \cX \times \cY \times \cZ$.
Then,
\begin{enumerate}
\item\label{it:Root} For the root $\lambda$ of $\cT$, it holds that $M^\pi_t[x,y] = 1$ for all $(x,y) \in \cX \times \cY$.
\item\label{it:Node} For a node $t$ of $\cT$ with two children $t0$ and $t1$, it holds that $M^\pi_t = M^\pi_{t0} + M^\pi_{t1}$. Moreover, if $t$ is an Alice's node, then $\alpha_t(x) = \alpha_{t0}(x) + \alpha_{t1}(x)$ for all $x \in \cX$ and $\beta_t(y) = \beta_{t1}(y) = \beta_{t0}(y)$ for all $y \in \cY$. Analogously, if $t$ is a Bob's node, then $\beta_t(y) = \beta_{t0}(y) + \beta_{t1}(y)$ and $\alpha_t(x) = \alpha_{t0}(x) = \alpha_{t1}(x)$ for all $y \in \cY$ and $x \in \cX$.
\item\label{it:Leaf} For a leaf $\ell$ of $\cT$ outputting $z \in \cZ$, it holds that for each $(x,y)$ with $M^\pi_\ell[x,y] \neq 0$ we have $(x,y,z) \in f$.
\item\label{it:Gamma} For each node $t$ of $\cT$, it holds that $\gamma_2(M^\pi_t) \leq 1$.
\end{enumerate}
\end{theorem}
\begin{proof}
Items~\ref{it:Root},~\ref{it:Node} and~\ref{it:Leaf} immediately follow from the fact that the collections $\alpha$ and $\beta$ satisfy the root, nodes, and computational constraints introduced in the last section.

Item~\ref{it:Gamma} can be shown by induction from the root $\lambda$.
By Item~\ref{it:Root}, we have that $\gamma_2(M^\pi_\lambda) = 1$.
By Item~\ref{it:Node}, we can easily verify that for any node $t$ with children $t0$ and $t1$ it holds that $\gamma_2(M^\pi_{t0}) + \gamma_2(M^\pi_{t1}) \leq \gamma_2(M^\pi_t)$.
\end{proof}

We end this section with a lower bound for $\gamma_2$ protocols.
For a relation $f \subseteq \cX \times \cY \times \cZ$, let $\gamma_2$ leaf complexity $\GLCC(f)$ denote the smallest number of leaves of the protocol structure of a $\gamma_2$ protocol that computes $f$. 
It clearly holds that
\[
  \GDCC(f) \ge \log \GLCC(f).
\]
We will show the following lower bound analogous to the rank lower bound in communication complexity.

\begin{theorem} 
\label{thm:LeafLB}
For any Boolean function $f:\cX\times\cY\to\ZO$, it holds that
  \[
    \gamma_2(f) \le \GLCC(f).
  \]
\end{theorem}
First, we prove an auxiliary lemma.
\begin{lemma}
\label{lem:LeafSum}
Let $\pi = (\cX \times \cY, \cT, d, \alpha, \beta)$ be a $\gamma_2$ protocol and $\cL$ be the set of leaves of $\cT$.
Then for each $(x,y) \in \cX \times \cY$, it holds that
  \[
    \sum_{\ell \in \cL} \langle \alpha_\ell(x), \beta_\ell(y) \rangle = 1.
  \]
\end{lemma}
\begin{proof}
We prove this by induction on the structure $\cT$. 
Fix a pair $(x,y) \in \cX \times \cY$.
At the root $\lambda$ of $\cT$, the root constraints give us
  \[
    \langle \alpha_\lambda(x), \beta_\lambda(y) \rangle = 1.
  \]
  Now, suppose that
  \[
    \sum_{\ell \in \cL'} \langle \alpha_\ell(x), \beta_\ell(y) \rangle = 1
  \]
  for a set $\cL'$ of nodes of $\cT$ containing an internal node $t$.
  Suppose that $t$ is an is an Alice's node (the other case is analogous).
  Then for the children $t0$ and $t1$ of $t$, we have
  \begin{align*}
    \langle \alpha_t(x), \beta_t(y) \rangle & = \langle \alpha_{t 0}(x), \beta_{t 0}(y) \rangle + \langle \alpha_{t 1}(x), \beta_{t 1}(y) \rangle.
  \end{align*}
  Here, we used that $\alpha_t(x) = \alpha_{t0}(x) + \alpha_{t1}(x)$ and $\beta_t(y) = \beta_{t0}(y) = \beta_{t1}(y)$ at Alice's nodes.
  Now, for the set $\cL'' = \cL' \setminus \{t\} \cup \{t0,t1\}$, it still holds that
  \[
  \sum_{\ell \in \cL''} \langle \alpha_\ell(x), \beta_\ell(y) \rangle = 1
  \] 
  The lemma is proven by proceeding in this way until there are only leaves left.

\end{proof}

\begin{remark*}
We remark that Lemma~\ref{lem:LeafSum} holds more generally for $\cL$ being any maximal antichain of $\cT$, not only the set of leaves.
\end{remark*}

\begin{proof}[Proof of Theorem~\ref{thm:LeafLB}]
  Let $\pi = (\cX \times \cY, \cT, d, \alpha, \beta)$ be a $\gamma_2$ protocol computing $f$. 
  For a leaf $\ell = tc$ (i.e., the leaf $\ell$ outputs $c \in \{0,1\}$), it holds that  $\langle\alpha_\ell(x),\beta_\ell(y) \rangle = 0$ for any $(x,y) \in \cX \times \cY$ with $f(x,y) \neq c$ by the computational constraints. 
  Let $\cL_1$ be the set of leaves of $\cT$ outputting 1.  
  By Lemma~\ref{lem:LeafSum}, it follows that 
  \[
    \sum_{\ell \in \cL_1} \langle\alpha_\ell(x),\beta_\ell(y) \rangle = f(x,y).
  \]
  In other words, $M_f = \sum_{\ell \in \cL_1} M^\pi_\ell$.
  Thus, we have
  \begin{align*}
  \gamma_2(f) &= \gamma_2(M_f) \leq \sum_{\ell \in \cL_1} \gamma(M^\pi_\ell) \tag{by the triangle inequality} \\
  &\leq |\cL_1|  \tag{by Item~\ref{it:Gamma} of Theorem~\ref{thm:GammaStructure}}
  \end{align*}  
  and the theorem follows.
\end{proof}

It follows that discrepancy lower-bounds generalized communication complexity.

\begin{corollary}
  $\GDCC(f) \ge \log\frac{1}{\disc(f)} - O(1)$
\end{corollary}

\begin{proof}
  Let $\mu$ be a distribution over $\cX\times \cY$ under which $f$ has $\disc_\mu(f) = \disc(f)$. It has been known since Linial and Shraibman~\cite{LS09} that, up to constant factors, this is equivalent to saying that the matrix $(\mu \circ f^{\pm1})[x,y] = \mu(x,y) \cdot (-1)^{f(x,y)}$ has small $\gamma_2^\ast$ norm:
  \[
    \gamma_2^\ast(\mu \circ f^{\pm1}) = \Theta(\disc_\mu(f)),
  \]
  where $\gamma_2^\ast$ is the dual norm of $\gamma_2$, i.e., $\gamma_2^\ast(M) = \sup_{X: \gamma_2(X) \leq 1} \ip{M}{X}$.
  It follows that
  \[
    \langle M_f, \mu\circ f^{\pm 1} \rangle \le \gamma_2(f) \cdot \gamma_2^\ast(\mu \circ f^{\pm 1}).
  \]
  The left-hand side measures exactly the probability that $f(x,y) = 1$ under $\mu$, which we may assume is $\ge 1/2$ (otherwise negate $f$, this adds at most $1$ to $\gamma_2(f)$). It follows that
  \[
    \GLCC(f) \ge \gamma_2(f) \ge \frac{1}{2} \cdot  \frac{1}{\gamma_2^\ast(\mu\circ f^{\pm 1})} = \Theta\left(\frac{1}{\disc(f)}\right),
  \]
  and thus $\GDCC(f) \ge \log\frac{1}{\disc(f)} - O(1)$ as intended.
\end{proof}

It follows that the inner product in $\mathbb{F}_2$ and a random function have nearly maximal generalized communication complexity.

\begin{corollary}
  $\GDCC(\mathsf{IP}_n) \ge \frac{n}{2} - 1$.
\end{corollary}

\begin{corollary}
  $\GDCC(f) = \Omega(n)$ for a random $f:\ZO^n\times\ZO^n\to\ZO$.
\end{corollary}

\subsection{Upper Bound for Equality}
\label{ssec:EQGammaProtocol}

In this section, we design an efficient (constant length) $\gamma_2$ protocol for the equality function -- $\EQ_d: [d] \times [d] \to \{0,1\}$ with $\EQ(x,y) = 1$ if and only if $x = y$.

The protocol construction will appear very mysterious unless some words are said about how it was discovered. The first observation is that protocols can be \textit{symmetrized}. Namely, the constraints defining a $\gamma_2$ protocol are invariant under two kinds of symmetries: the \textit{protocol symmetries} derived from graph isomorphisms of the protocol tree, and the \textit{function symmetries} which permute different inputs of the function while leaving the communication matrix unchanged. Since the solution space of the protocol constraints is convex, we can always take the average of a given solution under all possible symmetries, and the outcome will still be a solution to the constraints, which furthermore is invariant under all such symmetries. So for example, the communication matrix of equality is symmetric under the action of permuting the rows and columns by the same permutation. If we take any $d\times d$ matrix and symmetrize it (take the average) under this action, we will obtain a matrix such that all the diagonal entries have the same value $a$ and the off-diagonal entries have the same value $b$. Having realized this, we were trying to prove a lower-bound for equality by restricting our attention to such symmetric solutions. As it turns out, a $d\times d$ Gram matrix which has only two values, $a$ on the diagonal and $b$ on the off-diagonal, must be the Gram matrix of a family $\{\alpha(1), \ldots, \alpha(d) \}$ of vectors of a very special kind: the vectors must always be of the form $\alpha(x) = a'\phi + b'\eta(x)$, where $\phi$ is a vector that does not depend on $x$, and the vectors $\{\eta(1), \ldots, \eta(d)\}$ are the vertices of a $d$-simplex. This realization allowed us to analyze symmetric $\gamma_2$ protocols (for computing equality) very systematically, covering all possible cases, by varying only the constants $a'$ and $b'$. We expected this systematic analysis to result in a lower-bound, but when trying to prove it, we instead realized that a protocol for solving equality arises by setting the constants appropriately. This discussion, at least, serves to explain why simplexes appear in the protocol.

The protocol will consist only of 2 rounds where Alice ``sends'' 1 of 11 messages and Bob ``replies'' with a 1-bit message (which will be the output of the protocol).
In the previous sections, we considered only $\gamma_2$ protocols where the players send only 1 bit in each round. 
When we are in an Alice's node $t$ of a protocol structure $\cT$ we infer that $\alpha_t(x) = \alpha_{t0}(x) + \alpha_{t1}(x)$ and $\ip{\alpha_{t0}(x)}{\alpha_{t1}(x)} = 0$ for all possible $x$.
In words, $\alpha_{t0}(x)$ and $\alpha_{t1}(x)$ form an orthogonal decomposition of $\alpha_t(x)$.
It is straightforward to generalize $\gamma_2$ protocols so that players can send longer messages in one node.
Say, in an Alice's node $t$, Alice can send 1 of $\ell$ different messages (i.e., the node $t$ has $\ell$ children in $\cT$), then for each possible $x$ we have $\ell$ vectors $\alpha_{t1}(x),\dots,\alpha_{t\ell}(x)$ that form an orthogonal decomposition of $\alpha_t(x)$, i.e., $\alpha_t(x) = \sum_{i \in [\ell]} \alpha_{ti}(x)$ and $\ip{\alpha_{ti}(x)}{\alpha_{tj}(x)} = 0$ for all different $i$ and $j$ in $[\ell]$.

Let $\cT_\ell$ be the following protocol structure:
\begin{enumerate}
    \item The root $\lambda$ of $\cT_\ell$ is an Alice's node and has degree $\ell$.
    \item Each child of the root is a Bob's node and has exactly 2 children.
\end{enumerate}

\begin{theorem}
    Let $\ell \geq 11$ and $d > 0$ be integers.
    Then, there is $\gamma_2$ protocol with the structure $\cT_\ell$ computing the equality function $\EQ_d$.
\end{theorem}
\begin{proof}
    We need to design vectors $\alpha_t(x)$ and $\beta_t(y)$ for all possible $x,y \in [d]$ and nodes $t$ of $\cT_\ell$.
    Let $\psi$ be a unit vector and we set $\alpha_\lambda(x) = \beta_\lambda(y) = \psi$ for all $x,y \in [d]$ to satisfy the root constraints.
    Now, we design the vectors for the children of the root $\lambda$.
    Let $m$ be a child of $\lambda$.
    Since the root $\lambda$ of $\cT_\ell$ is an Alice's node, we set $\beta_m(y) = \beta_\lambda(y) = \psi$.
    We define
    \[
    \alpha_m(x) = \frac{1}{\ell} \cdot \psi + c\cdot \phi_m + c\cdot  \rho_m \otimes \eta_x,
    \]
    where
    \begin{enumerate}
        \item $c = \frac{\sqrt{\ell - 1}}{\sqrt{2}\ell}$.
        \item The vectors $\phi_m$'s and $\rho_m$'s are vertices of $\ell$-simplex centered at zero, i.e., $\ip{\phi_m}{\phi_m} = 1$, $\ip{\phi_m}{\phi_{m'}} = -\frac{1}{\ell - 1}$ if $m \neq m'$, and $\sum_{m \in [\ell]} \phi_m = 0$.
        \item The vectors $\eta_x$'s are vertices of $d$-simplex centered at zero.
        \item Vectors $\psi$, $\phi_m$'s, $\rho_m$'s, and $\eta_x$'s are orthogonal to each other -- which can be easily achieved if the dimension of the vectors is large enough.
    \end{enumerate}
    \begin{claim}
    \label{cl:AliceDecomposition}
        For any $x \in [d]$, the vectors $\alpha_1(x),\dots,\alpha_\ell(x)$ form an orthogonal decomposition of $\alpha_\lambda(x)$.
    \end{claim}
    \begin{proof}
        Fix $x \in [d]$.
        \[
            \sum_{m \in [\ell]} \alpha_m(x) = \psi + c \cdot \sum_{m \in [\ell]} \phi_m + c \cdot \left(\sum_{m \in [\ell]} \rho_m\right) \otimes \eta_x = \psi
        \]
        The last inequality holds because $\phi_m$'s and $\rho_m$'s are vertices of $\ell$-simplexes, so they sum to zero.
        Let $m$ and $m'$ be two different messages in $\ell$.
        Then,
        \begin{align*}
        \ip{\alpha_m(x)}{\alpha_{m'}(x)} &= \frac{1}{\ell^2} \cdot \ip{\psi}{\psi} + c^2 \cdot \ip{\phi_m}{\phi_{m'}} + c^2 \cdot \ip{\rho_m}{\rho_{m'}} \cdot \ip{\eta_x}{\eta_x} 
        \tag{by the orthogonality of the vectors and properties of the tensor product} \\
        &= \frac{1}{\ell^2} - \frac{2\cdot c^2}{\ell - 1} \tag{by properties of simplexes and $\psi$ being a unit vector} \\
        &= 0. \tag{since $c = \frac{\sqrt{\ell - 1}}{\sqrt{2}\ell}$}        
        \end{align*}
    \end{proof}

    Now, fix a child $m$ of the root $\lambda$ of $\cT_\ell$ (i.e., an Alice's message).
    The node $m$ is a Bob's node and it has 2 children $m0$ and $m1$.
    Thus, we set $\alpha_{m0}(x) = \alpha_{m1}(x) = \alpha_m(x)$.
    Further, for $b \in \{0,1\}$ and $y \in [d]$, we set
    \[
    \beta_{mb}(y) = \frac{1}{2}\cdot \psi + (-1)^b p \cdot \phi_m - (-1)^b q \cdot \rho_m \otimes \eta_y + (-1)^b r \cdot \chi_y,
    \]
    where
    \begin{enumerate}
        \item The vectors $\phi_m$'s, $\rho_m$'s, and $\eta_y$'s are the same vectors that were used for the definition of $\alpha_m(x)$.
        \item The vectors $\chi_y$'s are unit vectors orthogonal to all other vectors.
        \item The coefficients $p$ and $q$ will be set later, however we will have that 
        \begin{equation}
         0 \leq p \leq \frac{1}{2}, \text{ and }  0 \leq q^2 \leq \frac{1}{4} - p^2. \label{ineq:CoefBound} \\   
        \end{equation}                 
        \item $r^2 = \frac{1}{4} - p^2 - q^2$. Note that $\frac{1}{4} - p^2 - q^2 \geq 0$ by (\ref{ineq:CoefBound}).
    \end{enumerate}
    \begin{claim}
    \label{cl:BobDecomposition}
     Let $m \in [\ell]$ and $y \in [d]$. Then, the vectors $\beta_{m0}(y)$ and $\beta_{m1}(y)$ form an orthogonal decomposition of $\beta_m(y)$.
    \end{claim}
    \begin{proof}
        \begin{align*}
        \beta_{m0}(y) + \beta_{m1}(y) &= \frac{1}{2}\cdot \psi + p \cdot \phi_m - q \cdot \rho_m \otimes \eta_y +  r \cdot \chi_y  \\
        &+ \frac{1}{2}\cdot \psi - p \cdot \phi_m + q \cdot \rho_m \otimes \eta_y -  r \cdot \chi_y    = \psi
        \end{align*}
        \begin{align*}
            \ip{\beta_{m0}(y)}{\beta_{m1}(y)} &= \frac{1}{4} \cdot \ip{\psi}{\psi} - p^2 \cdot \ip{\phi_m}{\phi_m} - q^2 \cdot \ip{\rho_m}{\rho_m} \cdot \ip{\eta_y}{\eta_y} - r^2\cdot \ip{\chi_y}{\chi_y} 
            \tag{by the orthogonality of the vectors and properties of the tensor product} \\
            &= \frac{1}{4} - p^2 - q^2 - r^2 \tag{by properties of simplexes and $\psi$ being a unit vector} \\
            &= 0 \tag{since $r^2 = \frac{1}{4} - p^2 - q^2$}
        \end{align*}
    \end{proof}
    By Claims~\ref{cl:AliceDecomposition} and~\ref{cl:BobDecomposition}, we have that the collections of vectors $\alpha$ and $\beta$ satisfy the Alice's and the Bob's constraints.
    It remains to prove that the $\gamma_2$ protocol $\pi = \bigl([d], [d], \cT_\ell, D, \alpha, \beta)$ (for an appropriate dimension $D$) computes the equality function $\EQ_d$.
    To prove the claim, we need to verify the collections $\alpha$ and $\beta$ satisfy the computational constraints.
    In particular, we need to show that for any first message $m \in [\ell]$, and any $x,y \in [d], x \neq y$ we have that
    \begin{equation}
        \ip{\alpha_{m1}(x)}{\beta_{m1}(y)} = 0, \text{ and } \ip{\alpha_{m0}(x)}{\beta_{m0}(x)} = 0. \label{eq:LeafConstraints}
    \end{equation}  
    We will show there is a setting of $p$ and $q$ satisfying the inequalities (\ref{ineq:CoefBound}) and the computational constraints (\ref{eq:LeafConstraints}).
    First, expand the inner products (\ref{eq:LeafConstraints}).
    \begin{align*}
        \ip{\alpha_{m1}(x)}{\beta_{m1}(y)} &= \ip{\alpha_{m}(x)}{\beta_{m1}(y)} = \frac{1}{2\ell} \cdot \ip{\psi}{\psi} - cp\cdot \ip{\phi_m}{\phi_m} + cq \cdot \ip{\rho_m}{\rho_m} \cdot \ip{\eta_x}{\eta_y} \\
        &= \frac{1}{2\ell} - cp - \frac{cq}{d-1} \\
        \ip{\alpha_{m0}(x)}{\beta_{m0}(x)} &= \ip{\alpha_{m}(x)}{\beta_{m0}(x)} = \frac{1}{2\ell} \cdot \ip{\psi}{\psi} + cp \cdot \ip{\phi_m}{\phi_m} - cq \cdot \ip{\rho_m}{\rho_m} \cdot \ip{\eta_x}{\eta_x} \\
        &= \frac{1}{2\ell} + cp - cq
    \end{align*}
    Thus by the computational constraints (\ref{eq:LeafConstraints}), we get the following system of linear equations (with variables $p$ and $q$).
    \begin{align*}
        \frac{1}{2\ell} - cp - \frac{cq}{d-1} &= 0 \\
        \frac{1}{2\ell} + cp - cq &= 0
    \end{align*}
    The solution of this system is
    \begin{align*}
    p &= \frac{1}{c\ell} \cdot \left(\frac{1}{2} - \frac{1}{d}\right) = \frac{\sqrt{2}}{\sqrt{\ell - 1}} \cdot \left(\frac{1}{2} - \frac{1}{d}\right), \\
    q &= \frac{1}{c\ell} \cdot \left(1 - \frac{1}{d}\right) = \frac{\sqrt{2}}{\sqrt{\ell - 1}} \cdot \left(1 - \frac{1}{d}\right).
    \end{align*}
    Since $\ell \geq 11$ (and we may assume $d \geq 2$ so that the function $\EQ_d$ is non-trivial), we have the following bounds.
    \begin{align*}
        0 &\leq p^2 \leq \frac{1}{2(\ell - 1)} \leq \frac{1}{20} \\
        0 &\leq q^2 \leq \frac{2}{\ell - 1} \leq \frac{2}{10} \leq \frac{1}{4} - p^2
    \end{align*}
    Thus, the constraints (\ref{ineq:CoefBound}) are satisfied by our setting of $p$ and $q$ and we conclude the proof.
\end{proof}

%% file: 04_quantumlab.tex
\section{Quantum Lab Protocols}\label{sec:quantumlab}

In the next section we describe the HQFP, and we show that its PSD relaxation results in the functional description given in the introduction.

\subsection{Definition of the Model}

We now define a \emph{deterministic protocol} as a tuple $\pi = (\cX\times\cY, \cT, C)$, where $\cX\times\cY$ is a finite product set of \emph{inputs}, $\cT$ is a protocol structure, and $C$ is a collection of maps $C_t:\cX\times\cY\to\R$, satisfying the following constraints.

\begin{description}
\item[Root constraints.] For the root $\lambda$ of $\cT$ we have:
\begin{align*}
  C_\lambda(x,y) \cdot C_\lambda(x',y') & = 1 & \forall x,x' \in \cX, y,y'\in\cY
\end{align*}
These imply that the values $C_\lambda(x,y)$ are either all $1$, or all $-1$.

\item[Alice's nodes constraints.] Let $t \in \cT$ be an Alice node with two children $t0, t1$. (Think that Alice sends a bit $i$ to Bob when going into $t i$.) We impose the following constraints.
\begin{align*}
  C_{t0}(x,y)^2 + C_{t1}(x,y)^2 & = C_{t}(x,y)^2 & \forall x \in \cX\\
  C_{t0}(x,y) \cdot C_{t}(x,y) + C_{t1}(x,y) \cdot C_{t}(x,y) & = C_{t}(x,y)^2 & \forall x \in \cX\\
  C_{t0}(x,y) \cdot C_{t1}(x,y') & = 0 & \forall x \in \cX, y, y' \in \cY\\  
\end{align*}

Take these constraints together. 
By hypothesis, we assume that $C_{t}(x,y)\in \{0,\pm 1\}$, moreover the signs of every non-zero $C_{t}(x,y)$ are the same.
From the third constraint we conclude that at least one of $C_{t0}(x,y)$ or $C_{t1}(x,y')$ is zero. Together with the first two constraints, it then follows that for each $x$ we must choose either $C_{t0}(x,y) = C_{t}(x,y)$ and $C_{t1}(x,y) = 0$ for all $y$, or $C_{t1}(x,y) = C_{t}(x,y)$ and $C_{t0}(x,y) = 0$ for all $y$. Thus, if $C_{t}$ is the indicator of a rectangle $ A\times B \subseteq \cX\times\cY$ (which is the case at the root node) then $C_{ti}$ are indicators of two disjoint rectangles $A_i \times B$. This is the usual definition of a protocol.

\item[Bob's nodes constraints.] The constraints for Bob's nodes are analogous to Alice's node constraints.
\end{description}

Seeing that the above is a HQFP, we then relax it to a SDFP, replacing scalars with vectors and products with inner products. A \emph{deterministic quantum-lab protocol}, then, is a tuple $\pi = (\cX\times\cY, \cT, d, \psi)$, where $\cX\times\cY$ is a finite product set of \emph{inputs}, $\cT$ is a protocol structure, and $\psi$ is a collections of maps $\psi_t:\cX\times\cY\to\R^d$, for each node $t \in \cT$, satisfying the following constraints.

\begin{description}
\item[Root constraints.] For the root $\lambda$ of $\cT$ we have:
\begin{align*}
  \psi_\lambda(x,y) \cdot \psi_\lambda(x',y') & = 1 & \forall x,x' \in \cX, y,y'\in\cY
\end{align*}
This implies that every $\psi_\lambda(x,y)$ is the same unit-length vector.

\item[Alice's nodes constraints.] For $t \in \cT$ an Alice node with children $t0, t1$:
\begin{align*}
  \|\psi_{t0}(x,y)\|^2 + \|\psi_{t1}(x,y)\|^2 & = \|\psi_{t}(x,y)\|^2 & \forall x \in \cX\\
  \langle \psi_{t0}(x,y), \psi_{t}(x,y) \rangle + \langle \psi_{t1}(x,y), \psi_{t}(x,y)\rangle & = \|\psi_{t}(x,y)\|^2 & \forall x \in \cX\\
  \langle\psi_{t0}(x,y), \psi_{t1}(x,y')\rangle & = 0 & \forall x \in \cX, y, y' \in \cY\\  
\end{align*}
We will analyze these constraints just below.

\item[Bob's nodes constraints.] The constraints for Bob's nodes are analogous to Alice's node constraints.
\end{description}

How to interpret the above semidefinite program? Let us think of each $\psi_t(x,y)$ as an (unnormalized) quantum state. Then the root constraints say that the initial state, at the root $\lambda$, is the same for all $(x,y)$. The constraints at an Alice node say that $\psi_{t0}(x,y)$ and $\psi_{t1}(x,y)$ are an orthogonal decomposition of $\psi_t(x,y)$, but furthermore every quantum state $\psi_{t0}(x,y)$ is orthogonal to every $\psi_{t1}(x,y')$. This implies that there exists a pair of orthogonal projections $\Pi_{t,x,0}, \Pi_{t,x,1}$ such that $\psi_{ti}(x,y) = \Pi_{t,x,i} \psi_t(x,y)$ (e.g. $\Pi_{t,x,0}$ projects onto the span of every $(\psi_{t0}(x,y))_{y\in \cY}$, and $\Pi_{t,x,1}$ projectos to its orthogonal complement). In other words, to each $t$ and each $x$ corresponds a measurement, and $\psi_{ti}(x,y)$ is the (unnormalized) state obtained by measuring $\psi_t(x,y)$. Likewise, the constraints at Bob's nodes are equivalent to the existence of such a measurement $(\Pi_{y,0}, \Pi_{y,1})$ depending only on $t$ and $y$. This is precisely the definition of quantum lab protocols given in the introduction to this section.

We are only missing the constraints that define when a protocol computes a relation. So let $f \subseteq \cX\times\cY\times\cZ$ be a relation with output set $\cZ \subseteq \ZO^k$ and let $\pi = (\cX\times\cY, \cT, d, \psi)$ be a quantum lab protocol. 
We say that \emph{$\pi$ computes $f$} if the depth of every leaf $\ell \in \cT$ is at least $k$, and $\psi$ satisfies:
\begin{description}
    \item[Computational constraints.]  For every leaf $\ell \in \cT$ of the form $\ell = t z$ for some $z \in \ZO^k$ we have the following constraints:
    \begin{align*}
    \|\gamma_\ell(x,y)\|^2 & = 0 & \forall (x,y) \in \cX \times \cY \text{ s.t. } (x,y,z) \notin f
\end{align*}
\end{description}

It is then seen that deterministic protocols are quantum lab protocols with the constraint $d = 1$. In this case, the computational constraints imply that the last $k$ bits of communication are always a valid answer to the relation. Let us define the \textit{quantum-lab complexity} of a relation is the smallest depth of a quantum lab protocol that computes $f$. We can now ask what is the complexity of functions and relations in this model.

\subsection{A 2-round Lower Bound for Equality}\label{sec:quantum-lab-2-round-lower-bound}

Our first result is a two-round lower-bound. We will show that the equality function on $n$ bits needs $\Omega(n)$ bits to be computed by a two-round quantum lab protocol, i.e., a quantum lab protocol where Alice speaks, and then Bob speaks, with his last measurement giving the answer.

Indeed, if Alice has input $x$ and makes $k$ measurements, then the initial state $\psi_\lambda$ is broken into an orthogonal decomposition, which does not depend on $y$ since Bob did not speak yet:
\[
    \psi_\lambda = \sum_{t} \psi_{t}(x,y) = \sum_t \psi_t(x) \qquad \langle \psi_t(x), \psi_{t'}(x) \rangle = 0
\]
Now, if $2^k < 2^n$, the QPHP (\Cref{thm:qphp}) states that there must exist some message $t$, and two inputs $x, x'$, such that
\[
\langle \psi_t(x), \psi_t(x') \rangle \neq 0.
\]
Now Bob comes along and does some measurements. Suppose he has input $x$. Since $\psi_t(x)$ and $\psi_t(x')$ are not orthogonal, then no matter which measurement he does, there must be an outcome $i$ such that $\psi_{ti}(x,x)$ and $\psi_{t i}(x',x)$ are both non-zero. It follows that $ti$ is not monochromatic, i.e., the computational constraints associated with leaf $t i$ are not obeyed.

\subsection{Model Collapse -- All Functions Are Easy}\label{sec:quantum-lab-model-collapse}

\begin{theorem}
    Given any function $f: \cX \times \cY \to \ZO$, there is a 3-round Quantum Lab protocol using 4 bits of communication that computes $f$.
\end{theorem}

\begin{proof}
In our protocol given below the root node is a Bob node, The nodes at depths 1 and 2 are Alice nodes, the nodes at depth 3 are Bob nodes and the depth 4 nodes are leaves. We refer to nodes using their partial transcripts (i.e. elements of $\ZO^{\leq 4}$ with $\eps$ being the empty string). We refer to the state in the quantum lab at a node $v$ on inputs $x$ and $y$ as $\ket{\psi^{xy}_v}$.

The state in the quantum lab has 3 registers, which we number $1'$, $2'$ and $3$. Register $3$ is 2-dimensional with basis states $\ket{0}$ and $\ket{1}$ (i.e. the register consists of one qubit) and registers $1'$ and $2'$ are $|\cX|+|\cY|+1$-dimensional with their basis states being $\ket{\bot}$, $\ket{x}$ and $\ket{y}$ for each $x \in \cX$ and $y \in \cY$. We now provide the states in the quantum lab at each node for the first three bits of communication.

\begin{itemize}
    \item $\ket{\psi_{\eps}^{xy}} = \ket{0}_3\ket{\bot}_{1'}\ket{\bot}_{2'}$
    \item $\ket{\psi_0^{xy}} = \frac{1}{2} \ket{0}_3(\ket{\bot}_{1'}+\ket{y}_{1'})\ket{\bot}_{2'}$
    \item[] $\ket{\psi_1^{xy}} = \frac{1}{2} \ket{0}_3(\ket{\bot}_{1'}-\ket{y}_{1'})\ket{\bot}_{2'}$
    \item $\ket{\psi_{00}^{xy}} = \frac{1}{4} \ket{0}_3(\ket{\bot}_{1'}+\ket{y}_{1'})(\ket{\bot}_{2'}+\ket{x}_{2'})$
    \item[] $\ket{\psi_{01}^{xy}} = \frac{1}{4} \ket{0}_3(\ket{\bot}_{1'}+\ket{y}_{1'})(\ket{\bot}_{2'}-\ket{x}_{2'})$
    \item $\ket{\psi_{000}^{xy}} = \frac{1}{2}\left( \frac{1}{4} \ket{0}_3(\ket{\bot}_{1'}+\ket{y}_{1'})(\ket{\bot}_{2'}+\ket{x}_{2'}) + \frac{1}{2\sqrt{2}} \ket{1}_3(\ket{x}_{1'}+(-1)^{f(x,y)}\ket{y}_{1'})\ket{\bot}_{2'} \right)$
    \item[] $\ket{\psi_{001}^{xy}} = \frac{1}{2}\left( \frac{1}{4} \ket{0}_3(\ket{\bot}_{1'}+\ket{y}_{1'})(\ket{\bot}_{2'}+\ket{x}_{2'}) - \frac{1}{2\sqrt{2}} \ket{1}_3(\ket{x}_{1'}+(-1)^{f(x,y)}\ket{y}_{1'})\ket{\bot}_{2'} \right)$
\end{itemize}

We will address the last bit of communication after analyzing the above. We have only specified the relevant states along the all-$0$ transcript, and we will show that these can be realized by a quantum lab protocol. The states that appear along the other transcripts are the same up to some sign changes and so can also be realized similarly. As an example of how the states differ along different transcripts, here is the state at a node of depth 3:
\begin{align*}
\ket{\psi^{xy}_{b_1b_2b_3}} = \frac{1}{2}\Big( &\frac{1}{4} \ket{0}_3(\ket{\bot}_{1'}+(-1)^{b_1}\ket{y}_{1'})(\ket{\bot}_{2'}+(-1)^{b_2}\ket{x}_{2'})\\
& + (-1)^{b_3} \frac{1}{2\sqrt{2}} \ket{1}_3(\ket{x}_{1'}+(-1)^{b_1}(-1)^{f(x,y)}\ket{y}_{1'})\ket{\bot}_{2'} \Big)
\end{align*}

To show that the above quantum states can be realized by a quantum lab protocol, we will verify that the quantum lab protocol constraints are satisfied by these. For each node $v \in \{\eps,0,00\}$ it suffices to verify the following.

\begin{itemize}
    \item $\psi_{v}^{xy} = \psi_{v0}^{xy}+\psi_{v1}^{xy}$.\\
    This constraint is easy to verify.
    \item At an Alice node $v$, $\braket{\psi^{xy}_{v0},\psi^{xy'}_{v1}} = 0$ for all $x,y,y'$.\\
    This constraint is easy to verify for $v = 0$. For $v=00$, this inner product is
    \[ \frac{1}{4}\left(\frac{1}{16} \cdot 1 \cdot (1 + [y=y']) \cdot 2 - \frac{1}{8} \cdot 1 \cdot (1 + (-1)^{f(x,y) + f(x,y')}[y=y']) \cdot 1\right)\] where $[y=y']$ is $1$ if $y=y'$ and $0$ otherwise. Note that this is $0$ both when $y \neq y'$ and when $y = y'$.
    \item At a Bob node $v$, $\braket{\psi^{xy}_{v0},\psi^{x'y}_{v1}} = 0$ for all $x,x',y$\\
    Since the only Bob node in the first three bits is $\eps$, we only need to ensure that $\braket{\psi^{xy}_{0},\psi^{x'y}_{1}} = 0$. This is again easy to verify.
\end{itemize}

\subsubsection*{The final bit of communication}
We now make an additional observation about the state that we have reached after 3 bits of communication. Namely, fix any $y \in \ZO^n$ and let $x,x'$ be two inputs such that $f(x,y) \neq f(x',y)$. Then \[ \braket{\psi^{xy}_{000},\psi^{x'y}_{000}} = \frac{1}{4} \left( \frac{1}{16} \cdot 1 \cdot 2 \cdot 1 + \frac{1}{8} \cdot 1 \cdot (- 1) \cdot 1 \right) = 0. \] As a consequence $V^y_0 := \lspan(\{\psi^{xy}_{000}\}_{x : f(x,y) = 0})$ is orthogonal to $V^y_1 := \lspan(\{\psi^{xy}_{000}\}_{x : f(x,y) = 1})$. So now Bob can perform the measurement $\{\Pi_{V^y_0}, I - \Pi_{V^y_0}\}$. The output of the measurement is the value of $f(x,y)$.
\end{proof}

%% file: 05_nogo.tex
\section{A no-go theorem}\label{sec:nogo}

In the context of our work, the Sum-of-Squares (SoS) framework deals with a finite system of multivariate polynomial equations $\{p(x)=0\}_{p \in P}$ over a set of real variables $x$.\footnote{More generally, SoS allows for polynomial inequalities $q(x) \ge 0$, but we won't use them so we simplify the discussion by ignoring this possibility.} If this system is not satisfiable, the Positivstellensatz guarantees \cite[Section 6.4]{krajivcek2019proof} that there exists an element $p'$ in the ideal generated by $P$ and a sum-of-squares polynomial $q$ such that $p'+q=-1$. Here the ideal generated by $P$ refers to the set of polynomials $p'(x)$ such that $p'(x) = \sum g_i(x) p_i(x)$ for arbitrary polynomials $g_i$ and each $p_i \in P$. A sum-of-squares polynomial is a polynomial $q(x)$ such that $q(x) = \sum h_i(x)^2$ for arbitrary polynomials $h_i(x)$. The existence of such $g_i$ and $h_i$ refutes the satisfiability of the system of equations, because any solution $x$ of the system would give $p'(x) = 0$ and $q(x) \geq 0$, so $p' + q = -1$ would be impossible. The polynomials $g_i$ and $h_i$ together form what is called a \textit{Sum-of-Squares proof}, and the \textit{degree} of the proof is the maximum degree of any $g_i$ or $h_i$.

\subsection{HQFPs, SDFPs, and SoS proofs}

In a SDFP we are asked whether there exists a PSD matrix $K \in \R^{N \times N}$ whose entries satisfy some linear equations. Using $\vec{K} \in \R^{N^2}$ to denote the vector of entries of $K$ under a fixed ordering of $[N] \times [N]$, the previous sentence asks whether there is a PSD $K$ such that $V\vec{K} = a$ for some given $V \in \R^{m \times N^2}$ and $a \in \R^m$.

In order to view this in the SoS framework, we must rephrase these as polynomial equations over some variables. To this end, let $\{x_{i,j}\}_{i,j \in [N]}$ be our set of variables. Let $X$ denote the $N \times N$ matrix whose $i,j$th entry is $x_{i,j}$. Since $N \times N$ PSD matrices are exactly those that can be written as $M^TM$ for some $N \times N$ matrix $M$, the SDFP is equivalent to asking whether there is a setting of the variables $x$ such that $X^TX$ satisfies the linear inequalities $VX^TX = a$. To be more explicit we have one constraint for each $i \in [m]$, namely that the following quadratic form must evaluate to 0. \[ \sum_{j_1,j_2 \in [N]} V_{i,(j_1,j_2)} \sum_{k \in [N]} x_{k,j_1} x_{k,j_2} - a_i = 0. \]

This is our system of polynomial equalities, and it is satisfiable if and only if the SDFP is feasibly.

\medskip\noindent
We now show that if the SDFP is \textit{not} satisfiable, and its simple dual is satisfiable (which is guaranteed, e.g., under the Berman--Ben-Israel criterion---\Cref{thm:ben-israel}), there always exists a degree-2 SoS proof that the system of polynomial equations above is \textit{not} satisfiable.
 
A solution to the simple dual of a SDFP is a vector $w \in \R^m$ such that $V^Tw = \vec{M}$ for some PSD matrix $M$ and $w^Ta < 0$. The existence of $w$ proves that the SDFP is infeasible. Indeed, the linear equations $V\vec{K} = a$ imply that $\braket{\vec{M},\vec{K}} = w^TV\vec{K} = w^Ta < 0$, and yet $M$ is a PSD matrix, so $\braket{\vec{M}, \vec{P}}$ is non-negative for any other PSD matrix $P$, hence any solution to the linear equations cannot be PSD.

We now use such a solution $w$ to the simple dual to construct an SoS proof. In fact there exists a \emph{linear} combination of the polynomials plus a sum of squares polynomial that will simplify to a negative constant, proving that there is no assignment satisfying all the polynomial equations. To start with, consider the linear combination of constraints $-w^T(VX^TX-a)$, or more explicitly

\begin{align*}
    &\sum_{i \in [m]} - w_i \left(\sum_{j_1,j_2 \in [N]} V_{i,(j_1,j_2)} \sum_{k \in [N]} x_{k,j_1} x_{k,j_2} - a_i \right)\\
    = &-\sum_{j_1,j_2 \in [N]} M_{j_1,j_2} \sum_{k \in [N]} x_{k,j_1} x_{k,j_2} + \sum_{i \in [m]} w_i a_i.
\end{align*}

We know that the second term is a negative number. We also know that $M$ is PSD, so it can be written as $Y^TY$ for some $Y \in \R^{N \times N}$. Hence

\begin{align*}
    \sum_{j_1,j_2 \in [N]} M_{j_1,j_2} \sum_{k \in [N]} x_{k,j_1} x_{k,j_2} &= \sum_{j_1,j_2 \in [N]} \left(\sum_{k \in [N]} Y_{k,j_1} Y_{k,j_2}\right) \left(\sum_{k \in [N]} x_{k,j_1} x_{k,j_2}\right)\\
    &= \sum_{j_1,j_2,k_1,k_2 \in [N]} Y_{k_1,j_1} Y_{k_1,j_2} x_{k_2,j_1} x_{k_2,j_2}\\
    &= \sum_{k_1,k_2 \in [N]} \left(\sum_{j_1 \in [N]} Y_{k_1,j_1} x_{k_2,j_1}\right) \left(\sum_{j_2 \in [N]} Y_{k_1,j_2} x_{k_2,j_2}\right)\\
    &= \sum_{k_1,k_2 \in [N]} \left(\sum_{j \in [N]} Y_{k_1,j} x_{k_2,j}\right)^2.
\end{align*}

By adding the above sum of squares of linear forms (our $q$) to the linear combination of the quadratic forms (our $p'$), we are left with a negative number and so we have our degree-2 SoS proof.

\subsection{The no-go theorem: upper bounds on semidefinite relaxations of communication complexity follow from lower bounds on SoS degree}\label{subsec:implicitub}

\newcommand{\SCircSOS}[2]{\mathsf{SCircuit}_{#1}(#2)}
\newcommand{\CommSOS}[2]{\mathsf{Comm}_{#1}(#2)}
\newcommand{\CircSOS}[2]{\mathsf{Circuit}_{#1}(#2)}
\newcommand{\Out}{\mathsf{Out}}

Here we show implicit upper bounds on our variants of communication complexity for Karchmer--Wigderson relations. These upper bounds follow from proving the lack of lower bounds, which we can prove via the above connection to SoS proofs. To do so, we take the SDFP for any of our communication variants. We rephrase the SDFP as a polynomial system of equations as shown in the previous subsection. We then add more polynomial constraints, ensuring that every variable is Boolean (with constraints of the form $x_{i,j}^2-x_{i,j}=0$) and further ensuring that only variables of the form $x_{k,1}$ can be non-zero. This ensures that each vector is really just a Boolean value. This undoes the semidefinite relaxation and ensures that the problem is now just the HQFP capturing a deterministic communication protocol. Note that by the nature of the HQFP, the structure of the protocol (who speaks at what node) is fixed.

Although these constraints cannot be represented in the form of a modified SDFP, it is still a system of quadratic equations. As shown in the previous subsection, a dual solution to the original SDFP gives a degree-2 SoS refutation of the system of quadratic equations we get from the SDFP. Hence the same proof is also a refutation of the additionally constrained equations, which are equivalent to the HQFP. The rest of this section focuses on showing a lower bound on the SoS degree of refuting the HQFP. If the degree lower bound is larger than 2, this will show that there is no dual solution to the SDFP, which implies that it must have a primal solution (an upper bound). We build on the recent work of Austrin and Risse~\cite{austrin2023sum} to show the degree lower bound.

\subsubsection*{The SoS degree lower bound for circuit size lower bounds}

Austrin and Risse showed that there is \emph{no} low-degree SoS proof proving that a truth table is not computable by a small circuit! That is, given the truth table of a function $f : \ZO^n \rightarrow \ZO$ and a natural number $s$, they put forth a system of polynomial equations $\CircSOS{s}{f}$ such that The polynomial equations are simultaneously satisfiable if and only if there is a circuit of size $s$ that computes the function $f$. They then showed the following.

\begin{theorem}[\cite{austrin2023sum}]
    For all $\epsilon > 0$ there is a $\gamma$ such that: For all $n \in \mathbb{N}$, $s \geq n^{\gamma}$, and $f: \ZO^n \rightarrow \ZO$, the SoS degree of refuting $\CircSOS{s}{f}$ is $\Omega_{\epsilon}(s^{1-\epsilon})$.
\end{theorem}

Note that when $\CircSOS{s}{f}$ is satisfiable, the SoS degree is thought of as infinity since you can't refute it. Their result is nearly optimal since they also show that for unsatisfiable instances there is an SoS refutation of it in degree $O(s)$.

Their proof works via reduction to another SoS degree lower bound on refuting parity-CSPs. Given a bipartite graph $G = (U,V,E)$ and a string $f \in \{0,1\}^U$, consider the system of parity constraints $\mathrm{Parity_{G,f}}$ defined as $\{ \oplus_{v \in N(u)} y_v = f_u \mid u \in U \}$. Note that if $|V| < |U|$, there are strings $f$ for which this is not satisfiable. The following lower bound is implicit in a classic paper of Dima Grigoriev.

\begin{theorem}[\cite{grigoriev2001paritysos}]\label{thm:paritycspsoslb}
    Let $G = (U,V,E)$ be a bipartite expander graph with left-degree $k$ and with the property that for every subset $S \subseteq U$ of size at most $r$, the size of the neighbourhood of $S$, $|N(S)|$, is at least $2|S|$. Then for any $f \in \ZO^U$ the SoS degree of refuting $\mathrm{Parity_{G,f}}$ is at least $\Omega(r)$.
\end{theorem}

Austrin and Risse reduce $\CircSOS{s}{f}$ to this by cleverly restricting the circuit so that it necessarily computes a truth table of a specific form. They take an explicit expander graph $G$ as mentioned above, with $U=\ZO^n$. With $V = [m]$, they take $m$ gates of the circuit and treat them as unset constants. These are referred to as $y_1,\dots,y_m$. Since $G$ is explicit, they can entirely restrict the rest of the circuit so that it maps the input $u \in \ZO^n$ to $\oplus_{v \in N(u)} y_v$. The truth table of the circuit is then one of $2^{m}$ possibilities, and refuting $\CircSOS{s}{f}$ under this restriction is the same as refuting the parity-CSP $\mathrm{Parity}_{G,f}$. The authors show that this equivalence holds even within the SoS framework, and so the lower bound of $\Omega(r)$ on the SoS degree of $\mathrm{Parity}_{G,f}$ also applies to the restricted version of $\CircSOS{s}{f}$. Since their restriction blows up the degree by a factor of $k$, a degree lower bound of $\Omega(r/k)$ is shown on refuting the unrestricted $\CircSOS{s}{f}$. This can be made $\Omega_{\epsilon}(s^{1-\epsilon})$ by carefully choosing an explicit expander.

\subsubsection*{Our SoS degree lower bounds for proving Karchmer--Wigderson game lower bounds}

It is not clear how a communication problem can embed a parity-CSP. However, a communication protocol for the Karchmer--Wigderson relation of $f$ is equivalent to a formula for computing $f$. By doing this reduction in the SoS framework, we can hope to then use the refutation-degree lower-bound for $\CircSOS{s}{f}$ to show a refutation-degree lower-bound against our HQFP, as we encoded it above. Since we care about the depth of the communication protocol, the Karchmer--Wigderson correspondence will give us a circuit of small depth. Additionally since our communication model assumes a fixed structure to the communication program, we need to fix the structure of the circuit as well before embedding it. Since we will eventually be using the lower bound achieved by embedding a parity-CSP in a circuit, we need to ensure that the reduction to the parity-CSP can also happen using small-depth circuits with a fixed structure. Henceforth, we assume some familiarity with the paper of Austring and Risse.

We will be working with a simplified version of the $\CircSOS{s}{f}$ program. The original program is very flexible allowing one to create a circuit of any shape, choosing the gates at each node. We will deal with more standardized circuits. Any depth-$d$ circuit with the standard AND/OR/NOT gates and whose leaves are from $x_1,\dots,x_n,0,1$ can be expanded and written as a depth-$d$ formula with AND/OR gates whose leaves are from $x_1,\dots,x_n,\neg x_1, \dots, \neg x_n, 0, 1$ with the help of de Morgan's laws. We can then rewrite this as a complete binary tree of depth $2d+1$ with alternating layers of AND and OR gates by duplicating subformulas as required. Finally we can replace each leaf with a copy of an Indexing gadget: that is, a formula for Indexing that indexes into the set $x_1,\dots,x_n,\neg x_1, \dots, \neg x_n, 0, 1$. By setting the indices appropriately in a gadget, we can make the output gate of that gadget carry the same bit as the leaf value that was originally there. We will also require that the gadget be such that the indexing bits can be partially set in such a way that one unset bit will choose between the gadget outputting the constant $0$ or $1$. (This is to allow the unset bits to be treated as the $y_i$ variables, as done in~\cite{austrin2023sum}.) This can be done, adding $\log n + O(1)$ to the depth. Note that (a) every depth-$d$ circuit can be rewritten as such by just choosing the appropriate indexing bits in each copy, (b) the structure of the circuit is fixed and so the program does not have to concern itself with figuring out what gate a node has and which nodes it connects to, and (c) we can label a node with the path, or ``transcipt'' taken to get there from the root. Austrin and Risse's lower bound method continues to work. That is, one can take their restricted circuit that is equivalent to the parity-CSP instance and convert it to the above form. Extra restrictions will be needed so that each copy of an unset constant corresponding to a specific unset bit $y_i$ (from the original circuit) is set in the same way. The resulting system of polynomial equations would still be equivalent to the parity-CSP system of polynomial equations, even within the SoS framework, and the degree blowup they experience with these restrictions is still a multiplicative factor of $k$.

Let us call this simplified system $\SCircSOS{d}{f}$. Since the structure is fixed, all its variables and axioms are just ensuring that it represents the computation of a circuit that computes the function $f$.

\begin{definition}
    The variables of $\SCircSOS{d}{f}$ are $\Out_x(t)$ for each input $x$ and each node $t \in \ZO^{\leq d}$ (referred to with the transcript that specifies the node), denoting the bit computed at node $t$ on input $x$.

    The axioms state that the following polynomials must equal $0$.
    \begin{itemize}
    \item $\Out_x(t)(1-\Out_x(t)) = 0$ for all $x,t$. These ensure that all values are Boolean.
    \item $\Out_x(\lambda) - f(x) = 0$ for each input $x$ where $\lambda$ denotes the root node.
    \item $\Out_x(t) - \Out_x(t0) \Out_x(t1) = 0$ for each node $t$ which has an AND gate and children $t0$ and $t1$.
    \item $(1-\Out_x(t)) - (1-\Out_x(t0))(1-\Out_x(t1)) = 0$ for each node $t$ which has an OR gate and children $t0$ and $t1$.
    \item For leaves $\ell$ feeding in the $i$th bit of the input, we have $\Out_x(\ell) - x_i = 0$. For negated inputs, we would replace $x_i$ with $1-x_i$, and for constants we would replace it with the constant.
\end{itemize}
\end{definition}

Let us now take an HQFP $\CommSOS{d}{f}$ that states that there is a cost-$d$ communication protocol for the Karchmer--Wigderson relation of a function $f$, where the communication protocol structure matches that of a simplified circuit, so that
\begin{itemize}
    \item Alice speaks at nodes where there is an OR gate,
    \item Bob speaks at nodes where there is an AND gate, and
    \item at leaves where the input is $x_i$ or $\neg x_i$, the protocol outputs $i$.
\end{itemize}
It is easy to see that \emph{any} cost $d$ protocol can be modified to a cost $2d + \log n + O(1)$ protocol with such a structure, so up to a small change in parameters restricting ourselves in this way does not restrict the communication protocols under consideration. Hence proving lower bounds against all structures of an HQFP implies lower bounds against simplified structures of HQFPs and vice versa. To keep the same terminology, let us call HQFPs with this structure \emph{simplified HQFPs}.

Now given a simplified circuit $C$ that computes $f$, the Karchmer--Wigderson equivalence gives us a communication protocol that computes $KW_f$. It follows that there is a way to set the values of the variables in $\CommSOS{d}{f}$ as functions of the variables in $\SCircSOS{d}{f}$. We can write these functions as polynomials. For the HQFPs that we used, the polynomials mapping the variables were of low-degree. Now consider the system $\CommSOS{d}{f}$ with each variable replaced with the polynomials. It now has the same variables as $\SCircSOS{d}{f}$. Let us call this system $\mathsf{Sub}
\CommSOS{d}{f}$. If the axioms of $\mathsf{Sub}
\CommSOS{d}{f}$ follow from the axioms of $\SCircSOS{d}{f}$ in a ``low-degree'' derivation, then refuting $\mathsf{Sub}
\CommSOS{d}{f}$ also refutes $\SCircSOS{d}{f}$. Again, for the HQFPs we used, there were such low-degree derivations. Indeed, one imagines that a HQFP must be quite contrived in order for this not to be the case. This motivates the following definition.

\begin{definition}
    Let $\overline{w}$ denote the variables of $\SCircSOS{d}{f}$. A simplified HQFP $\CommSOS{d}{f}$ is $(c_1,c_2)$-contrived if there exist degree-$c_1$ real polynomials $p_x(\overline{w})$ for each variable $x$ in $\CommSOS{d}{f}$ such that:
    \begin{enumerate}
        \item For any setting of $\overline{w}$ satisfying the axioms of $\SCircSOS{d}{f}$, setting variables $x$ to  $p_x(\overline{w})$ satisfies the axioms of $\CommSOS{d}{f}$.
        \item Every axiom of the system $\mathsf{Sub}\CommSOS{d}{f}$ (obtained by replacing $x$ with $p_x(\overline{w})$) is in the ideal of the axioms of $\SCircSOS{d}{f}$ with only a $c_2$ increase in degree. That is, for every $p$ such that $p(\overline{w})=0$ is an axiom of $\mathsf{Sub}\CommSOS{d}{f}$, we can write $p = \sum h_i q_i$ where $\max_i{deg(h_i q_i)} - deg(p) \leq c_2$, and for each $q_i$, $q_i(\overline{w})=0$ is an axiom of $\SCircSOS{d}{f}$.
    \end{enumerate}
    (Note that the first point is actually redundant since the substitution in the second point ensures that the axioms of $\SCircSOS{d}{f}$ holding implies that the axioms of $\mathsf{Sub}\CommSOS{d}{f}$ holds.) 
\end{definition}

In the language of proof complexity, the above definition merely states that a HQBF is contrived with small parameters if it has an efficient Polynomial Calculus reduction to $\SCircSOS{d}{f}$.

Both our $\gamma_2$ protocols and our Quantum Lab protocols are relaxations of HQFPs that are $(d,O(d))$-contrived. We provide a proof below for $\gamma_2$ protocols, but omit the proof for Quantum Lab protocols as it is very similar (and perhaps a bit simpler).

But before that, let us prove a no-go theorem for proving communication lower bounds on SDFPs that are relaxations of HQFPs that are not highly contrived. 

\begin{theorem}[No-go Theorem]\label{thm:nogo}
    Let $f: \ZO^n \to \ZO$, and take $d \geq \log^c n$ for a large enough constant $c$. Let $\CommSOS{d}{f}$ be a simplified HQFP that formalizes communication complexity, and which is $(c_1,c_2)$-contrived, with $c_1, c_2 = o(n)$. Then, if the SDFP relaxation of $\CommSOS{d}{f}$ obeys the Berman--Ben-Israel criterion, there must exist a protocol for $KW_f$.
\end{theorem}

The theorem should be interpreted as saying: Either we have formalized communication complexity in a weird way, i.e. using a very contrived HQFP, or by a formalization whose relaxation does not obey the Berman--Ben-Israel criterion, or otherwise the resulting PSD relaxation can solve every KW game in depth $d$. As we will now see in the proof, the depth $d$ is the smallest depth of a circuit that can enumerate the neighbours of a sufficiently good bipartite expander, with $2^n$ nodes on the left and $2^{(\log n)^3}$ nodes on the right. After a brief search, the best explicit construction we could find \cite{tashma2001expander} gives us $\log^c n$ depth, so that is how we stated the theorem above, but we expect it should be possible to construct such expanders in $\NC_1$, in which case the no-go theorem can be improved improved to $d = O(\log n)$. 

\begin{proof}
    We start by considering a degree $\alpha$ SoS refutation of our HQFP. This will naturally give an SoS refutation of $\mathsf{Sub}\CommSOS{d}{f}$ as well. Note that the refutation with the substituted variables is a refutation of degree $c_1\alpha$. Now by the definition of contrived, the SoS refutation of $\mathsf{Sub}\CommSOS{d}{f}$ is in fact a refutation of $\SCircSOS{d}{f}$ since we can replace every axiom of the former with an element in the ideal of the latter. This only increases the degree by at most $c_2$. Hence we get a refutation of $\SCircSOS{d}{f}$ of degree $c_1\alpha + c_2$. Proving that such a refutation requires degree larger than $2c_1+c_2$ will be sufficient, since it implies that $\alpha > 2$. This is because if the SDFP relaxation has a solution to its dual, the HQFP must have a refutation of degree $2$ (see the discussion at the beginning of Section~\ref{subsec:implicitub}). Since the dual to the SDFP would not be satisfiable, the primal must be satisfiable. That is, the SDFP has a protocol of depth $d$.

    The question now is, along the lines of \cite{austrin2023sum}, how large an expander $G$ can we have while embedding the system $\mathrm{Parity_{G,f}}$ inside a depth $d$ simplified circuit? We are trying to maximize the value of $r/k$ (see Theorem~\cite{grigoriev2001paritysos} for the definitions of $r,k$) to ensure that it is $\omega(2c_1 + c_2)$. This is because the SoS degree lower bound on the parity-CSP is $\Omega(r)$, which will translate to a lower bound of $\Omega(r/k)$ for $\SCircSOS{d}{f}$. We know from~\cite{austrin2023sum} that an embedding like above works for the circuit size program, we just need to modify it to work for circuit depth.

    For our embedding we start by modifying the construction in the proof of \cite[Lemma 25]{austrin2023sum}. We change their explicit choice of the $m$-bit parity portion of their circuit to have depth $\log m$ instead. This does not change their proof. We also have to ensure that their Selector circuits are low-depth. This depends on the explicitness of the expander they use. By using the expander mentioned in~\cite[Theorem 3]{tashma2001expander} (building on the condenser from~\cite{raz1999condenser}), the selector is implemented in depth polylogarithmic in its input size. As a reminder of the notation for the parameters of the expander, $|U| = 2^n, |V| = m$, the left degree is $k$ and sets of size up to $r$ expand by a factor of $2$. The expander can achieve, for any $\alpha>0$, the parameters $k = \mathrm{poly}(n)$, $m = 2^{(\log r)^{1+\alpha}}$. The depth of the circuit is then $d = \log n + \log m + \mathrm{polylog}(n+\log m)$. We can set $r = kn$ satisfying the above constraints, making $d = \mathrm{polylog}(n)$ as well. Finally, we make their circuit into the form of a simplified circuit as we mentioned earlier in the section, doubling the depth and then adding $\log n$ more to the depth. The depth will remain $\mathrm{polylog}(n)$. This will give us a lower bound of $\Omega(r/k) = \Omega(n)$, which implies that $\alpha \geq 
    \Omega((n - c_2)/c_1) \geq \omega(1)$, thus finishing the proof.
\end{proof}

\subsection{\texorpdfstring{$\gamma_2$}{gamma-2} protocols are not ``weird''}

We have already seen that the SDFPs defining $\gamma_2$ protocols obey the Berman--Ben-Israel criterion. We now show that the SDFP defining $\gamma_2$ protocols is not very contrived. The same can be shown for quantum-lab protocols. Intuitively, for a formalization of communication complexity to be $(c_1, c_2)$-contrived for small $c_1,c_2$, it should suffice that the variables which define the protocol that solves a Karchmer--Wigderson game of $f$ depend on few of the variables that define the formula for solving $f$. This is a kind of ``locality'' constraint seems to hold for all arguments where one shows that one kind of algorithm is simulating another. In principle there could be exceptions to this rule, but we cannot think of any.

\begin{theorem}
    The SDFP defining $\gamma_2$ protocols is a PSD relaxation of an HQFP that is $(d,O(d))$-contrived.
\end{theorem}

\begin{proof}
    As a reminder, the variables of the HQFP are $A_{t}(x)$ and $B_{t}(y)$ for each node $t \in \{0,1\}^{\leq d}$ in the protocol tree (again referred to using the transcript), every $x \in f^{-1}(1)$ and $y \in f^{-1}(0)$. These denote whether or not the rectangle at node $t$ contains $x$ and $y$. We use the following polynomials to set these values:
    \begin{itemize}
        \item $A_{\lambda}(x) = B_{\lambda}(y) = 1$. The rest of the values are recursively defined as follows.
        \item For an Alice node $t$ (corresponding to an OR gate in the circuit), define $B_{t0}(y) = B_{t1}(y) = B_{t}(y)$. Also define $A_{t0}(x) = A_{t}(x) \Out_x(t0)$ and $A_{t1}(x) = A_{t}(x) (1-\Out_x(t0))$. Note that if $\Out_x(t) = 1$, this ensures that Alice goes to its left-most child that also outputs $1$.
        \item For a Bob node $t$ (corresponding to an AND gate in the circuit), define $A_{t0}(x) = A_{t1}(x) = A_{t}(x)$. Also define $B_{t0}(y) = B_{t}(y) (1-\Out_y(t0))$ and $B_{t1}(y) = B_{t}(y) \Out_y(t0)$. Note that if $\Out_y(t) = 0$, this ensures that Alice goes to its left-most child that also outputs $0$.
    \end{itemize}

    Note that the variables $A_{t}(x)$ and $B_{t}(y)$ are defined recursively and so as a polynomial in the $\Out$ variables they would have degree comparable to (and at most) the depth of the node $t$, which is at most $d$.

    We now move on to the second point of the definition of contrived. We will consider the axioms of $\mathsf{Sub}\CommSOS{d}{f}$. Before we dive into it, note that given a Boolean axiom $x^2 -x = 0$ and a monomial $m$ of the form $m'x^2$, we can write $m$ as $(x^2-x)m' + xm'$. In this fashion, we can remove any degrees larger than $1$ without any degree increase in the proof.
    \begin{itemize}
        \item $A_{\lambda}(x)A_{\lambda}(x') - 1$ is just $0$ once the substitution is done. All the root constraints are of the same sort.
        \item At an Alice node $t$, we have the axiom $A_{t0}(x) A_{t1}(x)$ that must be $0$. Expanding the substitution by one step, we see that it is $A_{t}(x)^2 \Out_x(t0)(1-\Out_x(t0))$, which is in the ideal of the Boolean axiom of $\Out_x{t0}$. 
        \item At an Alice node $t$, we also have the axiom $A_{t0}(x)^2 + A_{t1}(x)^2 - A_{t}(x)^2$ must be $0$. Using the Boolean axioms of the $\Out$ variables, this reduces to asking whether $A_{t0}(x) + A_{t1}(x) - A_{t}(x)$ is in the ideal. This is true by definition since the latter is $A_{t}(x)(1-1)$ which is $0$.
        \item The final axiom at an Alice node $t$ is $A_{t0}(x)A_{t}(x) + A_{t1}(x)A_{t}(x) - A_t(x)^2$. Again, this simplifies to $0$ after expanding the substitution by one step.
        
        Bob's corresponding axioms are handled similarly. None of these showed any increase in the degree. This leaves us with the leaf axioms.
        \item For all $x,y$ such that $x_i=y_i$ and for all leaf nodes $t$ that are labeled $i$, we have the axiom $A_{t}(x)B_{t}(y)$ must be $0$. To show this, we translate the proof of the KW reduction to the language of polynomials. As intermediate steps, we want to show the following polynomials are in the ideal.
        \begin{itemize}
            \item $A_{t}(x)(1-\Out_x(t))$. This being $0$ means that if Alice reaches a node $t$, that node in the circuit must output $1$.
            \item $B_{t}(y)\Out_y(t)$. This being $0$ means that if Bob reaches a node $t$, that node in the circuit must output $0$.
        \end{itemize}
        Assume we've shown the above are in the ideal. At a leaf node $t$ labeled $i$, we have the axioms $\Out_x(t) - x_i = 0$ and $\Out_y(t) - y_i = 0$. So if $x_i = y_i$, $\Out_x(t) - \Out_y(t)$ is in the ideal. Hence the polynomial $B_{t}(y) (A_{t}(x)(1-\Out_x(t))) + A_{t}(x) (B_{t}(y)\Out_y(t)) - A_{t}(x)B_{t}(y)(\Out_x(t) - \Out_y(t))$, which simplifies to $A_{t}(x)B_{t}(y)$, is also in the ideal. This last step does incur a degree increase of $1$.
    
        We prove that the needed polynomials are in the ideal by induction. It is true at the root because we have the circuit axioms $\Out_x(\lambda) - 1$ and $\Out_y(\lambda)$ must be $0$. We give the induction step at a Bob node, the Alice node case is similar. At a Bob node $t$ (which is an AND gate in the circuit),
        \begin{itemize}
            \item $B_{t0}(y)\Out_y(t0) = B_{t}(y)(1-\Out_y(t0))\Out_y(t0)$ which is in the ideal of the Boolean axioms.
            \item $B_{t1}(y)\Out_y(t1) = B_{t}(y) \Out_y(t0) \Out_y(t1)$ which can be derived from $B_{t}(y) \Out_y(t)$ (which is in the ideal by induction) using the axiom $\Out_y(t0) \Out_y(t1) - \Out_y(t)$.
            \item $A_{t0}(x)(1-\Out_x(t0)) = A_{t}(x)(1-\Out_x(t0))$. This in turn happens to be equal to
            \begin{align*}
                &A_{t}(x)(1-\Out_x(t0)\Out_x(t1))(1+\Out_x(t0)\Out_x(t1)-\Out_x(t0))\\
                - &A_{t}(x) \Out_x(t0)^2 \Out_x(t1)(1-\Out_x(t1)).
            \end{align*} $A_{t}(x)(1-\Out_x(t0)\Out_x(t1))$ is in the ideal since we have the axiom $\Out_x(t0)\Out_x(t1) - \Out_x(t)$ and since $A_{t}(x)(1-\Out_x(t))$ is in the ideal by induction. Since we have the Boolean axiom for $\Out_x(t1)$, the ideal also includes the above expression, which is $A_{t}(x)(1-\Out_x(t))$. A similar derivation holds for $A_{t1}(x)(1-\Out_x(t1))$. The degree has increased by $3$ in this inductive step.
        \end{itemize}
        The overall degree increase after all the induction and the final step is $O(d)$.
    \end{itemize}
\end{proof}

%% file: 06_conclusion.tex

%% file: main.bbl
\newcommand{\etalchar}[1]{$^{#1}$}
\begin{thebibliography}{HJKP10}

\bibitem[AGG12]{polyhedra_equiv_mean_payoff}
M.~Akian, S.~Gaubert, and A.~Guterman.
\newblock Tropical polyhedra are equivalent to mean payoff games.
\newblock {\em International Journal of Algebra and Computation}, 22(1), 2012.

\bibitem[AGS18]{issac2016jsc}
X.~Allamigeon, S.~Gaubert, and M.~Skomra.
\newblock Solving generic nonarchimedean semidefinite programs using stochastic
  game algorithms.
\newblock {\em Journal of Symbolic Computation}, 85:25--54, 2018.

\bibitem[AR05]{aharonov2005lattice}
Dorit Aharonov and Oded Regev.
\newblock Lattice problems in {$\NP \cap \coNP$}.
\newblock {\em Journal of the ACM}, 52(5):749--765, 2005.

\bibitem[AR23]{austrin2023sum}
Per Austrin and Kilian Risse.
\newblock Sum-of-squares lower bounds for the minimum circuit size problem.
\newblock In {\em Proceedings of CCC}, 2023.

\bibitem[BBI71]{berman_more_1971}
Abraham Berman and Adi Ben-Israel.
\newblock More on linear inequalities with applications to matrix theory.
\newblock {\em Journal of Mathematical Analysis and Applications},
  33(3):482--496, 1971.

\bibitem[BI69]{benisrael_linear_1969}
Adi Ben-Israel.
\newblock Linear equations and inequalities on finite dimensional, real or
  complex, vector spaces: A unified theory.
\newblock {\em Journal of Mathematical Analysis and Applications},
  27(2):367--389, 1969.

\bibitem[BT{\etalchar{+}}22]{bun2022approximate}
Mark Bun, Justin Thaler, et~al.
\newblock Approximate degree in classical and quantum computing.
\newblock {\em Foundations and Trends in Theoretical Computer Science},
  15(3-4):229--423, 2022.

\bibitem[Gri01]{grigoriev2001paritysos}
Dima Grigoriev.
\newblock Linear lower bound on degrees of positivstellensatz calculus proofs
  for the parity.
\newblock {\em Theoretical Computer Science}, 259(1):613--622, 2001.

\bibitem[HJKP10]{hrubevs2010convex}
Pavel Hrube{\v{s}}, Stasys Jukna, Alexander Kulikov, and Pavel Pudlak.
\newblock On convex complexity measures.
\newblock {\em Theoretical Computer Science}, 411(16-18):1842--1854, 2010.

\bibitem[Hå98]{stad1998shrinkage}
Johan Håstad.
\newblock The shrinkage exponent of de morgan formulas is 2.
\newblock {\em SIAM Journal on Computing}, 27(1):48--64, 1998.

\bibitem[KKN95]{karchmer1995fractional}
Mauricio Karchmer, Eyal Kushilevitz, and Noam Nisan.
\newblock Fractional covers and communication complexity.
\newblock {\em SIAM Journal on Discrete Mathematics}, 8(1):76--92, 1995.

\bibitem[KN97]{KN97}
Eyal Kushilevitz and Noam Nisan.
\newblock {\em Communication complexity}.
\newblock Cambridge University Press, 1997.

\bibitem[KP00]{khachiyan2000integer}
Leonid Khachiyan and Lorant Porkolab.
\newblock Integer optimization on convex semialgebraic sets.
\newblock {\em Discrete \& Computational Geometry}, 23(2):207--224, 2000.

\bibitem[Kra19]{krajivcek2019proof}
Jan Kraj{\'\i}{\v{c}}ek.
\newblock {\em Proof complexity}.
\newblock Cambridge University Press, 2019.

\bibitem[LMSS07]{LMSS07}
Nati Linial, Shahar Mendelson, Gideon Schechtman, and Adi Shraibman.
\newblock Complexity measures of sign matrices.
\newblock {\em Combinatorica}, 27(4):439--463, 2007.

\bibitem[LP18]{liu2018exact}
Minghui Liu and G{\'a}bor Pataki.
\newblock Exact duals and short certificates of infeasibility and weak
  infeasibility in conic linear programming.
\newblock {\em Mathematical Programming}, 167:435--480, 2018.

\bibitem[LP23]{lourencco2023simplified}
Bruno~F Louren{\c{c}}o and G{\'a}bor Pataki.
\newblock A simplified treatment of ramana’s exact dual for semidefinite
  programming.
\newblock {\em Optimization Letters}, 17(2):219--243, 2023.

\bibitem[LS{\etalchar{+}}09a]{lee2009lower}
Troy Lee, Adi Shraibman, et~al.
\newblock Lower bounds in communication complexity.
\newblock {\em Foundations and Trends in Theoretical Computer Science},
  3(4):263--399, 2009.

\bibitem[LS09b]{LS09}
Nati Linial and Adi Shraibman.
\newblock Learning complexity vs~communication complexity.
\newblock {\em Combinatorics, Probability and Computing}, 18(1–2):227–245,
  2009.

\bibitem[LS21]{li2021general}
Lily Li and Morgan Shirley.
\newblock The general adversary bound: A survey.
\newblock {\em arXiv preprint arXiv:2104.06380}, 2021.

\bibitem[Ram97]{ramana1997exact}
Motakuri~V. Ramana.
\newblock An exact duality theory for semidefinite programming and its
  complexity implications.
\newblock {\em Mathematical Programming}, 77:129--162, 1997.

\bibitem[Raz98]{razborov1998lower}
Alexander Razborov.
\newblock Lower bounds for the polynomial calculus.
\newblock {\em computational complexity}, 7(4):291--324, 1998.

\bibitem[RR97]{razborov1997natural}
Alexander~A Razborov and Steven Rudich.
\newblock Natural proofs.
\newblock {\em Journal of Computer and System Sciences}, 1(55):24--35, 1997.

\bibitem[RR99]{raz1999condenser}
Ran Raz and Omer Reingold.
\newblock On recycling the randomness of states in space bounded computation.
\newblock In {\em Proceedings of STOC}, 1999.

\bibitem[Rud97]{rudich1997super}
Steven Rudich.
\newblock Super-bits, demi-bits, and np/qpoly-natural proofs.
\newblock In {\em Proceedings of RANDOM/APPROX}, 1997.

\bibitem[Sha79]{shamir1979factoring}
Adi Shamir.
\newblock Factoring numbers in $o(\log n)$ arithmetic steps.
\newblock {\em Information Processing Letters}, 8(1):28--31, 1979.

\bibitem[Tou15]{touchette2015quantum}
Dave Touchette.
\newblock Quantum information complexity.
\newblock In {\em Proceedings of STOC}, 2015.

\bibitem[TSUZ01]{tashma2001expander}
Amnon Ta-Shma, Christopher Umans, and David Zuckerman.
\newblock Loss-less condensers, unbalanced expanders, and extractors.
\newblock In {\em Proceedings of STOC}, 2001.

\bibitem[TV08]{tarasov2008semidefinite}
Sergey~P Tarasov and Mikhail~N Vyalyi.
\newblock Semidefinite programming and arithmetic circuit evaluation.
\newblock {\em Discrete Applied Mathematics}, 156(11):2070--2078, 2008.

\end{thebibliography}
